\documentclass[12pt,centertags,reqno]{amsart}

\usepackage[foot]{amsaddr}
\usepackage{latexsym}
\usepackage[english]{babel}
\usepackage[T1]{fontenc}
\usepackage[numbers]{natbib}
\usepackage{amssymb}
\usepackage{fancyhdr}
\usepackage{url}
\usepackage{hyperref}
\usepackage{verbatim}
\usepackage{leftidx}
\usepackage{bm}
\usepackage{color,graphicx}

\usepackage{mathrsfs, mathtools}
\usepackage{stmaryrd}

\usepackage{marvosym}
\mathtoolsset{showonlyrefs}

\numberwithin{equation}{section} \makeatletter
\renewcommand{\subsection}{\@startsection
{subsection}{2}{0mm}{\baselineskip}{-0.25cm}
{\normalfont\normalsize\bf}} \makeatother

\newtheorem{theorem}{Theorem}[section]
\newtheorem{lemma}[theorem]{Lemma}
\newtheorem{corollary}[theorem]{Corollary
}
\newtheorem{definition}[theorem]{Definition}
\newtheorem{remark}[theorem]{Remark}
\newtheorem{proposition}[theorem]{Proposition}
\newtheorem{example}[theorem]{Example}
\newtheorem{assumption}[theorem]{Assumption}
\newtheorem{notation}[theorem]{Notation}

\def \A {\mathcal A}

\def \F {\mathcal F}

\def \P {\mathbf P}
\def \Q {\mathbf Q}
\def \I {{\mathbf 1}}

\def \R {\mathbb R}
\def \bF {\mathbb F}

\def \bE {\mathbb E}
\def \bN {\mathbb N}

\newcommand{\ud}{\mathrm d}
\newcommand{\ds}{\displaystyle}
\newcommand{\esp}[2][\mathbb E] {#1\left[#2\right]}

\def \a {{(1)}}
\def \b {{(2)}}
\def \i {{(i)}}

\begin{document}

\author[C.~Ceci]{Claudia  Ceci}
\address{Claudia  Ceci, Department of Economics,
University ``G. D'Annunzio'' of Chieti-Pescara, Viale Pindaro, 42,
I-65127 Pescara, Italy.}\email{c.ceci@unich.it}

\author[K.~Colaneri]{Katia Colaneri}
\address{Katia Colaneri, Department of Economics and Finance,
 University of Rome Tor Vergata, via Columbia, 2, I-00133 Rome, Italy}\email{katia.colaneri@uniroma2.it}

\author[A.~Cretarola]{Alessandra Cretarola}
\address{Alessandra Cretarola, Department of Mathematics and Computer Science,
 University of Perugia, Via Luigi Vanvitelli, 1, I-06123 Perugia, Italy.}\email{alessandra.cretarola@unipg.it}

\title[Optimal reinsurance and investment with common shock dependence]{Optimal reinsurance and investment under common shock dependence between financial and actuarial markets}

\date{}

\begin{abstract}
We study optimal proportional reinsurance and investment strategies for an insurance company which experiences both ordinary and catastrophic claims and wishes to maximize the expected exponential utility of its terminal wealth. We propose a model where the insurance framework is affected by environmental factors, and aggregate claims and stock prices are subject to common shocks, i.e. drastic events such as earthquakes, extreme weather conditions, or even pandemics,  that have an immediate impact on the financial market and simultaneously induce insurance claims. Using the classical stochastic control approach based on the Hamilton-Jacobi-Bellman equation, we provide a verification result for the value function via classical solutions to two backward partial differential equations and characterize the optimal reinsurance and investment strategies. Finally, we make a comparison analysis to discuss the effect of common shock dependence.
\end{abstract}

\maketitle




\section{Introduction}
The insurance business and the overall social economy are closely related and play crucial roles in the financial system. Basically, insurance promotes economic development and social stability via its own deepening and functions of transferring risks, financial intermediary, and loss compensation.
Indeed, insurance companies safeguard the financial stability of households and firms by insuring their risks.
Moreover, there are growing links between banks and insurance companies, 
which act as large investors and improve the economic efficiency of the system by spreading individual risks in the financial market. For instance, the policyholder transfers the risk to insurers by paying premiums and insurers employ premiums to make investments. \\
It is important to realize that, similarly to any other business, insurance companies require protection against risk. 
Reinsurance helps insurers to manage their risks by absorbing some of their losses. Legally, reinsurance is an insurance contract in which one insurance company indemnifies, for a premium, another insurance company against all or part of the losses that it may incur. 
Additionally, the management of the surplus of an insurance company may also involve investments in the financial market.
Optimal reinsurance and investment problems for various risk models have gained a lot of interest in financial and insurance literature. They have been widely investigated under different criteria, especially via expected utility maximization, mean-variance criterion and ruin probability minimization (see e.g. \citet{irgens2004optimal,liu2009optimal,shen2015optimal,gu2017optimal,brachetta2019proportional,cao2020optimal} and references therein).
Despite the richness of contributions, only a few papers on optimal investment-reinsurance address the problem in relation to dependent risks. Nevertheless, more and more natural and manmade disasters in recent years have brought great damage to the safety of
lives and properties for people and have demonstrated that the insurance businesses  and the financial market are not  independent to each other. Therefore,  considering dependent risk models in the actuarial literature has become necessary.\\
Pandemics, such as COVID-19, as well as other economically destructive natural phenomena such as hurricanes, earthquakes and wildfires, may have serious consequences for insurance companies, see e.g. some empirical studies as \citet{born2006catastrophic,benali2017impact,richter2020covid}. 
Reinsurance allows insurance companies to manage potential claim shocks by transferring a specific portion or category of risk out of their portfolios to reinsurers' portfolios, and hence, to stabilize their earnings and to increase their capacity to issue more business.\\
The main novelty of this paper is to consider an optimal investment-reinsurance problem for an insurance company
with two lines of business in a stochastic-factor model which accounts for a common shock dependence between the financial and the insurance markets. Precisely, the two lines of business correspond to ordinary claims and catastrophic claims, whose arrival intensities are affected by environmental (stochastic) factors, as well as the insurance and reinsurance premia; furthermore,
we assume that the arrival of catastrophic claims affects risky asset prices by inducing downward jumps. This modeling framework implies that the financial and the insurance frameworks are
correlated by means of a common shock.
Optimal reinsurance and investment problems in dependent financial and non-life insurance markets are also studied e.g. in \citet{hainaut2017contagion,brachetta2020optimal}. In \citet{hainaut2017contagion} a potential contagion risk between financial and insurance activities due to catastrophic events is modeled by time-changed processes; in \citet{brachetta2020optimal}, instead, a diffusion risk model with an external driver modeling a stochastic environment is considered. 
It is important to stress that in most of the contributions on optimal investment-reinsurance problems the expression {\em common shock dependence} refers to the assumption of dependent classes of insurance business, see e.g. \citet{yuen2015optimal, liang2016optimal2, han2019optimal}, where only the optimal proportional reinsurance problem is addressed, and see e.g. \citet{bi2016optimal,bi2019optimal,zhang2020optimal} which also include the optimal investment problem.
To the best of our knowledge, there are only a few papers, where the modeling framework allows for a common shock affecting both {\em the stock market and the insurance market}, see \citet{liang2016optimal,liang2018optimal}. However, in this literature, the insurance risk model is modulated by a compound Poisson process with
constant arrival intensity and risky asset jump sizes (induced by the common shock) are independent of the claim sizes; moreover, only classical premium calculation principles are considered, which considerably simplify the analysis. Instead, in our setting,
claim arrival intensities of both business lines are modeled as functions of additional exogenous stochastic factors; therefore, intensities of the claims arrival process are stochastic and evolve according to changing environmental conditions. Stochastic intensity models can be also found in non-life insurance settings by means of shot-noise Cox processes, see e.g. \citet{dassios2003pricing,schmidt2014catastrophe,cao2020optimal}, and more generally,
stochastic risk factor models in insurance are further considered in \citet{liang2014optimal,brachetta2019proportional,brachetta2019excess,brachetta2020bsde}.
To disentangle the effect of catastrophic and ordinary claims on the financial market we assume that they may produce jumps of different sizes on asset prices and these jumps depend on the claims size as well. In addition, we apply an extended version of the classical expected value and variance premium calculation principles. \\
The portfolio of the insurance company consists of the retention levels of proportional reinsurance agreements
and the investment strategy in the financial market, where a riskless asset and a risky asset are traded.
Using the classical stochastic control approach based on the Hamilton-Jacobi-Bellman equation, we provide a verification result for the value function (see Theorem \ref{Verification})
via classical solutions to two backward partial differential equations, one of which accounts for dependence between the insurance and the financial markets, and to the best of our knowledge, has not been previously analyzed in the literature. We characterize the optimal strategy (see Proposition \ref{MC}, Proposition \ref{existence} and Proposition \ref{ultima}) and discuss its properties by considering a comparison with a corresponding scenario without a common shock effect.\\
The paper is organized as follows. Section \ref{sec:model} introduces the mathematical framework for the financial-insurance market model. In Section \ref{sec:problem} we formulate the main assumptions and describe the optimization problem. Section \ref{sec:valuef} includes the derivation of the Hamilton-Jacobi-Bellman equation and the Verification Theorem. In Section \ref{sec:solution} we solve the resulting minimization problems, discussing in Section \ref{EVP} how the results apply under the expected value premium,
 and we characterize the optimal investment-reinsurance strategy in Section \ref{sec:strategy}. The effects of the common shock on the optimal strategy are investigated in Section \ref{sec:effects}. Moreover, some sufficient conditions for existence and uniqueness of the solution to the Hamilton-Jacobi-Bellman equation are provided in Section \ref{sec:existence}. Finally, most technical proofs are collected in Appendix \ref{app:proofs}.

\section{The mathematical model}\label{sec:model}

Let $(\Omega,\F, \P;\bF)$ be a filtered probability space and assume that the filtration $\bF=\{\F_t, \ t \in [0,T]\}$ satisfies the usual hypotheses. The time $T$ is a finite time horizon that may represent the maturity of a reinsurance contract.
We consider two independent counting processes $N^\a=\{N_t^\a, \ t \in [0,T]\}$, $N^\b=\{N_t^\b, \ t \in [0,T]\}$ that indicate the number of ordinary claims and catastrophic claims, respectively. In the sequel, we refer to ordinary and catastrophic claims as two business lines for the same insurance company. The jump times of $N^\a$ and $N^\b$ correspond to claim arrival times and are denoted by $\{T^\a_n\}_{n \in \bN}$ and $\{T^\b_n\}_{n \in \bN}$. We assume that the claim size of each line of business is described by two sequences of independent and identically distributed $\R^+$-valued random variables denoted by $\{Z^\a_n\}_{n \in \bN}$ and $\{Z^\b_n\}_{n \in \bN}$, respectively, with cumulative distribution functions $F^\a, F^\b:\R \to [0,1]$, satisfying the condition
$$\int_0^{+\infty} z^2 F^{(i)}(\ud z) < \infty,\ \quad i=1,2.$$
For any $n,m \in \mathbb{N}$, $Z^\a_n$ and $Z^\b_m$ indicate the claim amount at time $T^\a_n$ and $T^\b_m$ for each business line.
Processes $N^\a$,  $N^\b$, and the families of variables  $\{Z^\a_n\}_{n \in \bN}$ and $\{Z^\b_n\}_{n \in \bN}$ are assumed to be mutually independent.

We now introduce two stochastic processes $Y^\a=\{Y_t^\a, \ t \in [0,T]\}$, $Y^\b=\{Y_t^\b, \ t \in [0,T]\}$ representing {\em environmental factors} that affect arrival intensities of ordinary and catastrophic claims respectively. To model claim arrival intensities of the processes $N^\a$ and $N^\b$,  we consider two strictly positive measurable functions $\lambda^\i: [0,T] \times \R \to (0,+\infty)$, with $i=1,2$, and define the processes $\{\lambda^\i(t,Y_t^\i),\ t\in [0,T]\}$, for $i=1,2$, satisfying
\begin{equation}\label{eq:integrab_lambda}
\esp{\int_0^T \lambda^\i(t, Y^\i_t) \ud t} < \infty,\quad i=1,2.
\end{equation}
According to the classical construction of a Cox process we assume that for $i=1,2$,
\[
\mathbb{E}[N^\i_t-N^\i_s=k\mid \mathcal{F}^{Y^\i}_T\lor\mathcal{F}_s] = \frac{\bigl(\int_s^t \lambda^\i(u, Y^\i_u)\,\ud u\bigr)^k}{k!}e^{-\int_s^t \lambda^\i(u, Y^\i_u)\, \ud u}, \quad  0\leq s<t\leq T,
\]
for every $k=0,1, \dots$, where $\F^{Y^\i}_T:=\sigma\{Y^\i_t,\  0 \leq t \leq T\}$. Then, the process $\{\lambda^\i(t,Y_t^\i),\ t\in [0,T]\}$ provides the intensity of the counting process $N^\i$, for $i=1,2$, and condition \eqref{eq:integrab_lambda} guarantees that $N^\a$ and $N^\b$ are non-explosive (see \citet{bremaud1981}).

We define the cumulative claim process $C=\{C_t, \ t \in [0,T]\}$ as
\begin{equation}\label{loss}
C_t = C_t^\a+C_t^\b:=\sum_{j=1}^{N_t^\a}Z_j^\a+\sum_{j=1}^{N_t^\b}Z_j^\b,\quad t \in [0,T],
\end{equation}
where for any $t \in [0,T]$, $C_t^\a$ and $C_t^\b$ represent the cumulative claim amount corresponding to each line of business in the time interval $[0, t]$.

\begin{example}
Although we keep the parameters quite general, in our setting it is reasonable to assume that catastrophic claims occur less frequently then ordinary claims.  Moreover one can choose the parameters of the claim size distributions such that catastrophic events correspond to larger claims.
  An example of intensities that model this situation is
  \begin{align*}
  \lambda^\a(t, y^\a)&=\lambda(t) f_1(y^\a)\\
  \lambda^\b(t , y^\b)&=\lambda(t) f_2(y^\b)
  \end{align*}
  for functions $\lambda:[0,T] \to (0,+\infty)$, $f_1:\R\to (1,2)$ and $f_2:\R\to (0,1)$. We could also choose, for instance, that claim sizes of both types are exponentially distributed with distribution functions $F^\a(z)=1-e^{-a z}$ and $F^\b(z)=1-e^{-b z}$ with $b>a>0$, for all $z>0$.
\end{example}

We model the dynamics of the stochastic factors $Y^\a$, $Y^\b$ by diffusion processes
\begin{align}
\ud Y^\a_t= b^\a(t) \ud t + a^\a(t) \ud W^\a_t, \quad Y^\a_0=y^\a\in \R,\label{def:y1}\\
\ud Y^\b_t= b^\b(t) \ud t + a^\b(t) \ud W^\b_t, \quad Y^\b_0=y^\b\in \R,\label{def:y2}
\end{align}
where $a^\i, b^\i:[0,T] \to \R$, for  $i=1,2$, are measurable functions such that
\begin{align}\label{eq:a_i}
\int_0^T \left(\left(a^\i(t)\right)^2 + |b^\i(t)|\right) \ud t<\infty
\end{align}
and $W^\a=\{W_t^\a, \ t \in [0,T]\}$, $W^\b=\{W_t^\b, \ t \in [0,T]\}$ are independent Brownian motions also independent of $C^\a, C^\b$. 

In the sequel, we will make use of the notation $m^\a(\ud t, \ud z)$ and $m^\b(\ud t, \ud z)$ to indicate the random counting measures associated with $C^\a$ and $C^\b$ respectively, and which are defined by
\begin{align}
m^\i(\ud t, \ud z)=\sum_{n\in \bN} \delta_{\{T^\i_n, Z^\i_n\}}(\ud t, \ud z) \I_{\{T^\i_n\leq T\}}, \quad i=1,2,
\end{align}
where $\delta_{\{t,z\}}$ denotes the Dirac measure at point $(t,z)$. The dual predictable projections of the random measure are computed in the lemma below.

\begin{lemma}
Let $N^\a$ and $N^\b$ be Cox processes with stochastic intensities given by $\{\lambda^\a(t,Y_t^\a),\ t\in [0,T]\}$ and $\{\lambda^\b(t,Y_t^\b),\ t\in [0,T]\}$, respectively.
Then, the dual predictable projections of the measures $m^\i(\ud t, \ud z)$, $i=1,2$ are given by
\begin{align}
\nu^\i(\ud t, \ud z)=\lambda^\i(t, Y^\i_t)F^\i(\ud z)\ud t \quad i=1,2.
\end{align}
\end{lemma}

\begin{proof}
This is a consequence of the fact that for every nonnegative, predictable and $[0,+\infty)$-indexed process
$\{H(t,z), \ z \in  [0,+\infty), t\in[0,T]\}$, we have
\begin{equation}\label{n1}
\mathbb{E}\biggl[\int_0^T\int_0^{+\infty} H(t,z)\,m^\i(\ud t,\ud z)\biggr]=\mathbb{E}\biggl[\int_0^T\int_0^{+\infty}  H(t,z)\lambda^\i(t, Y^\i_t)F^\i(\ud z)\ud t \biggr], \quad i=1,2,
\end{equation}
see \citet[Theorem T3, Chapter VIII]{bremaud1981}.
\end{proof}

The insurance company invests premia received by the policyholders in a financial market where a riskless asset and a risky asset are negotiated and subscribes reinsurance contracts for each business line. Securities are traded continuously on the time interval $[0,T]$ and their price processes $B=\{B_t,\ t \in [0,T]\}$ and $P=\{P_t,\ t \in [0,T]\}$ follow the dynamics
\begin{equation}
\ud B_t = r(t) B_t \ud t, \quad B_0=1,
\end{equation}
where $r:[0,T] \to \R$ is an integrable function (i.e. $\int_0^T |r(s)| \ud s<\infty$) representing the interest rate, and
\begin{equation}\label{eq:P1}
\ud P_t = P_{t-}\left(\mu(t) \ud t + \sigma(t) \ud W_t - \int_0^{+\infty}K(t, z) m^\b(\ud t, \ud z)\right), \quad P_0 = p_0 >0,
\end{equation}
where $W=\{W_t,\ t \in [0,T]\}$ is a standard Brownian motion independent of $W^\a, W^\b, C^\a$ and $C^\b$, and the measurable functions $\mu:[0,T] \to \R$ and $\sigma:[0,T]\to (0, +\infty)$  are such that
\[
\int_0^T \left(|\mu(t)|+ \sigma^2(t)\right) \ud t <\infty,
\]
and $K:[0,T]\times [0,+\infty)\to [0,1)$. The latter, i.e. $K(t,z)\in [0,1)$,  ensures that $P$ stays positive. Moreover, because of condition \eqref{eq:integrab_lambda} and the fact that $K(t, z)\in [0,1)$, it holds that
\begin{align*}
&\bE\left[\int_0^T\int_0^{+\infty} K(t, z) \lambda^\b(t, Y^\b_t)F^\b(\ud z) \ud t\right]<\infty.
\end{align*}

The model for the price dynamics has the following interpretation: when a catastrophic claim arrives, it instantaneously affects the asset price which jumps downwards of an amount that depends on the size of the claim. Common shocks naturally induce dependence between the financial and the pure insurance frameworks.

\begin{notation}\label{notation:B}
  In the sequel we use the notation $B(t_1,t_2)=e^{\int_{t_1}^{t_2} r(s) \ud s}$, for $0\leq t_1<t_2\leq T$, to identify the accumulation factor, and hence $B(0,t)=B_t$ is the price of the riskless asset at time $t$. We also set $\overline{B}:=e^{\int_{0}^{T} |r(s)| \ud s}$.
\end{notation}


\section{The insurance and investment problem} \label{sec:problem}

The gross risk premium rate of each business line is  described by a nonnegative predictable process $\{c^\i(t, Y^\i_{t}), \ t \in [0,T] \}$, for $i=1,2$, where  $c^\i:[0,T]\times \R \to (0, + \infty)$ are measurable functions such that
\begin{equation}
\label{eqn:cpremium_int}
\mathbb{E}\left[\int_0^T c^\i(t, Y^\i_t)\, \ud t\right]< \infty,\quad i=1,2.
\end{equation}

To partially cover for the losses, the insurance company buys {\em proportional} reinsurance contracts with retention level $u\in[0,1]$, that is, retained losses are modeled via functions $g^\i(z,u)=zu$, for $i=1,2$,  where $1-u$ represents the percentage of losses covered by the reinsurance\footnote{It is well known that proportional reinsurance satisfies the properties of self-insurance functions (see, e.g. \citet[Chapter 4]{schmidli2007stochastic}).}.
We assume that the insurance company continuously buys reinsurance agreements whose reinsurance premium processes are defined as follows.

\begin{definition}[Reinsurance premia]\label{def:reinsurance_premium}
For $i = 1,2$, let the functions $q^\i:[0,T]\times \R \times [0,1] \to [0, +\infty)$ be such that $q^\i(t,y,u)$ is continuous with continuous first and second derivatives with respect to $u$ and such that
\begin{itemize}
\item[(i)] $q^\i(t,y,1)=0$, for every  $(t,y)\in[0,T]\times\R$, meaning that a null reinsurance is not expensive;
\item[(ii)] $q^\i(t,y,0)>c^\i(t, y)$, for each $(t,y)\in[0,T]\times\R$, which prevents the insurance company from making a risk-free profit;
\item[(iii)]  $\frac{\partial q^\i}{\partial u}(t,y,u)\leq 0$ for all $(t,y,u)\in [0,T]\times \R \times [0,1]$, since premia are decreasing with respect to the retention level (i.e. increasing with respect to the protection).
\end{itemize}

The derivatives $\frac{\partial q^\i}{\partial u}(t,y,0)$ and $\frac{\partial q^\i}{\partial u}(t,y,1)$ are interpreted as right and left derivatives, respectively.

Given a reinsurance strategy $\{u^\i_t,\ t \in [0,T]\}$ for $i=1,2$,  the associated reinsurance premium process is given by $\{q^\i(t, Y^\i_t, u^\i_t), \ t \in [0,T]\}$.
\end{definition}

We also assume
\begin{equation} \label{eqn:qpremium_int}
\mathbb{E}\biggl[\int_0^T q^\i(t,Y^\i_t, 0)\,\ud t\biggr]< \infty, \quad i=1,2,
\end{equation}
which means that the total expected premium payment in case of full reinsurance does not explode. 

\begin{example} \label{es:premi}
Here, we propose two forms for the insurance and the reinsurance premia which represent an extension to the stochastic case of the classical expected value and variance premium calculation principles. We suppose that $Z^\i$ denotes a random variable having distribution function $F^\i$, $ i=1,2$.
\begin{itemize}
\item[(1)] Under the {\em expected value principle}, we can choose $c^\i(t, Y^\i_t)$ and $q^\i(t, Y^\i_t, u^\i_t)$ as
\begin{align} \label{eq:expected_value_prop}
c^\i(t,y^\i)& =  (1+\theta^\i)\mathbb{E}[Z^\i]\lambda^\i(t,y^\i), \\
q^\i(t, y^\i, u^\i)& = (1+\theta_R^\i)\int_0^{+\infty} z(1 - u^\i_t) \lambda^\i(t,y^\i) F^\i(\ud z),
\end{align}
for $i=1,2,$ with $\theta_R^\i > \theta^\i>0$ being the safety loading of the reinsurance and insurance companies, respectively (see \citet{brachetta2019excess,brachetta2019proportional}).
\item[(2)] Under the {\em variance premium principle}, we can take $c^\i(t, Y^\i_t)$ and $q^\i(t, Y^\i_t, u^\i_t)$
with
\begin{align}
c^\i(t,y^\i) &=   \mathbb{E}[Z^\i]\lambda^\i(t,y^\i) + \theta^\i  \mathbb{E}[{Z^\i}^2] \lambda^\i(t,y^\i), \label{eq:variance_value1}\\
q^\i(t, y^\i, u^\i) &=  \int_0^{+\infty}  z(1 - u^\i_t) \lambda^\i(t,y^\i) F^\i(\ud z) + \theta_R^\i  \int_0^{+\infty}  z^2(1 - u^\i_t)^2 \lambda^\i(t,y^\i) F^\i(\ud z) \label{eq:variance_value2}
\end{align}
for $i=1,2$.
\end{itemize}
In both cases the expected cumulative premia are strictly greater than the expected cumulative losses. Indeed, from \eqref{n1}, for each $t \in [0,T]$ and for $i=1,2$
\begin{align}
\mathbb{E}\biggl[ \int_0^t c^\i(s,Y_s^\i) \ud s \biggr] > \mathbb{E}[ C^\i_t] = \mathbb{E}\left[ \int_0^t \int_0^{+\infty} z m^\i(\ud s, \ud z)\right] = \mathbb{E}[Z^\i]  \mathbb{E}\left[\int_0^t  \lambda^\i(s,Y^\i_s) \ud s \right]
\end{align}
and for any predictable reinsurance strategy  $\{u^\i_t, t \in [0,T]\}$ we have
\begin{align}
&\mathbb{E}\biggl[ \int_0^t q^\i(s, Y^\i_s, u^\i_s)\ud s \biggr] >\mathbb{E}\biggl[ \int_0^t \int_0^{+\infty} z(1-u^\i_s)\lambda^\i(s,Y^\i_s) F^\i(\ud z) \biggr]\\
&= \mathbb{E}\biggl[ \int_0^t \int_0^{+\infty} z(1-u^\i_s) m^\i(\ud s, \ud z) \biggr]  = \mathbb{E}\biggl[ \sum_{n=1}^{N_t^\i} Z_n^\i\left(1 - u^\i_{T_n^\i}\right)\biggr] ,
\end{align}
for every $t \in [0,T]$ and for $i=1,2$.
\end{example}

The total surplus (or total reserve) process is $R=\{R_t:=R^{\a,u^\a}_t+R^{\b,u^\b}_t, \ t \in [0,T]\}$, where for each $i=1,2$, $R^{\i, u^\i}$ is the surplus process associated to each line of business with a given reinsurance strategy $\{u^\i_t, \ t\in [0,T]\}$ satisfying
\begin{equation}
\label{eqn:surplus_process}
\ud R^{\i,u^\i}_t = \left[c^\i(t, Y^\i_t)-q^\i(t, Y^\i_t, u^\i_t)\right]\ud t - \int_0^{+\infty}zu_t^\i\,m^\i(\ud t,\ud z),
\end{equation}
and $R_0\in\mathbb{R}^+$ is the initial wealth of the insurance company.

Let $\{u^\i_t, \ t\in [0,T]\}$, for $i=1,2$, be the dynamic retention levels corresponding to each reinsurance contract and let  $\{w_t, \ t\in [0,T]\}$ be the total amount invested in the risky asset. We assume that for every $t \in [0,T]$, $w_t$ takes values in $\R$, which means that short-selling and borrowing from the bank account are allowed.  Then, the wealth of the insurance company associated with the reinsurance-investment strategy $\alpha=\{\alpha_t:= (u^\a_t, u^\b_t, w_t), \ t \in [0,T]\}$  satisfies
\begin{equation}
\label{eqn:wealth_proc}
\ud X^\alpha_t = \ud R^{\a,u^\a }_t+\ud R^{\b,u^\b}_t+ w_t \frac{\ud P_t}{P_t} + (X^\alpha_t - w_t) \frac{\ud B_t}{B_t}, \qquad X^\alpha_0=R_0\in\mathbb{R}^+.
\end{equation}

\begin{remark}
By expanding equation \eqref{eqn:wealth_proc} we get that
\begin{align}
\ud X^\alpha_t=& \mu^X_t \ud t +w_t \sigma(t) \ud W_t - \int_0^{+\infty} zu_t^\a m^\a(\ud t,\ud z) -   \int_0^{+\infty} \left(zu_t^\b + w_t K(t,z)\right) m^\b(\ud t,\ud z), \label{eq:wealth_2}
\end{align}
with $X^\alpha_0=R_0\in\mathbb{R}^+$, where  $ \mu^X_t = \mu^X(t, X^\alpha_t,Y^\a_t, Y^\b_t, \alpha_t)$ is given by
\begin{align}
& \mu^X(t, x,y^\a,y^\b, \alpha)  = \mu^X(t, x,y^\a,y^\b, u^\a, u^\b, w)\\
& \qquad :=c^\a(t, y^\a)+c^\b(t, y^\b)-\left(q^\a(t, y^\a, u^\a)+q^\b(t,y^\b, u^\b)\right)+w \left(\mu(t)-r(t)\right)+x r(t).
\end{align}
Moreover,  it is easy to verify that
\begin{align}
&X^\alpha_t =  R_0 B(0,t) + \int_0^t B(s,t)  \left( c^\a(s, Y^\a_s)+c^\b(s, Y^\b_s)- q^\a(s, Y^\a_s, u^\a_s)- q^\b(s,Y^\b_s, u^\b_s)\right) \ud s\\
&\quad +\int_0^t B(s,t) w_s \left(\mu(s)-r(s)\right)\ud s + \int_0^t{B(s,t) w_s\sigma(s)\,\ud W_s}
-\int_0^t\int_0^{+\infty} B(s,t) zu_s^\a \,m^\a(\ud s, \ud z) \\
& \quad - \int_0^t\int_0^{+\infty} B(s,t) [zu_s^\b + w_s K(s,z) ] \,m^\b(\ud s, \ud z), \label{eqn:wealth_sol}
\end{align}
for every $t \in [0,T]$, where $B(s,t)$, $0 \leq s < t \leq T$,  is defined in Notation \ref{notation:B}.
\end{remark}

The goal of the insurance company is to maximize the expected utility from its terminal wealth, that is, to solve the optimization problem
\begin{align}\label{pb:optimization}
\sup_{\alpha \in\mathcal{A}}{\mathbb{E}\bigl[U(X_T^\alpha)\bigr]},
\end{align}
where $U(x)=1 - e^{-\gamma x}$ represents the exponential preferences of the insurance company, with risk aversion parameter $\gamma>0$, and $\mathcal A$ denotes the class of admissible reinsurance-investment strategies defined as follows.


\begin{definition}\label{admi}
A strategy $\alpha=\{\alpha_t= (u^\a_t, u^\b_t, w_t), \ t \in [0,T]\}$ with values in $[0,1] \times [0,1] \times \R$ is said to be {\em admissible} if it is predictable and satisfies $\displaystyle \mathbb{E}\left[ e^{-\gamma X^\alpha_T} \right]< \infty$,
\begin{align}
&\mathbb{E} \left[\int_0^T \left(|w_s| \left(|\mu(s) - r(s)|+ \lambda^\b(s, Y^\b_s) \right) + w^2_s \sigma^2(s) \right)\ud s \right] < \infty,\label{eq:integrab_strategy}\\
&\mathbb{E} \left[\int_0^T e^{\gamma \overline{B} (z+|w_s|)} \lambda^\b(s, Y^\b_s) F^\b(\ud z) \ud s\right] < \infty,\label{eq:integrab_strategy2}
\end{align}
where $\overline{B}$ is defined in Notation \ref{notation:B}.
\end{definition}
Condition \eqref{eq:integrab_strategy2} guarantees that the Verification theorem (see Theorem \ref{Verification} below) applies. Nevertheless, this condition is satisfied by the optimal strategy under suitable assumptions on the model coefficients.

Now, we assume that the following assumptions are in force throughout the paper; they provide a set of sufficient conditions for a strategy $\alpha$ to be admissible (see Proposition \ref{ADM}).

\begin{assumption}\label{ASS}
\begin{itemize}
\item[]
\item[(i)]
There is an integrable function  $\delta : [0,T] \to (0,+ \infty)$ such that for $i=1,2,$
\begin{align}
q^\i(t, y^\i,0) \leq \delta(t) \quad    \quad \lambda^\i(t,y^\i) \leq \delta(t), \quad  \mbox{ for all } \  (t,y^\i) \in [0,T] \times \R. \end{align}
\item[(ii)] For $i=1,2,$ it holds that
\begin{align}
&\mathbb{E}\left[ e^{2\gamma Z^\i \overline{B}} \right]< \infty, \quad
\mathbb{E}\left[ {Z^\i}^2 e^{\gamma Z^\i \overline{B}} \right]< \infty.
\end{align}
\item[(iii)] There is a constant $C>0$ such that
$\frac{|\mu(t) - r(t)|}{\sigma(t)} \leq C$, for all $t \in [0,T]$.
\end{itemize}
\end{assumption}

\begin{proposition}\label{ADM}
Let $\alpha=\{\alpha_t= (u^\a_t, u^\b_t, w_t), \ t \in [0,T]\}$ be a predictable strategy with values in $[0,1] \times [0,1] \times \R$. Assume that 
there is a square integrable function $\eta: [0,T] \to (0,+ \infty)$ such that
\begin{align}\label{w_ass}
|w_t| \leq \eta(t),  \quad t  \in [0,T], \ \P-\mbox{a.s.},
\end{align}
and
\begin{gather}
\label{eq:delta1}
\int_0^T \Big(\eta(s) (|\mu(s) - r(s)|+ \delta(s))+ \eta^2(s) \sigma^2(s) \Big) \ud s <  \infty,\\
\label{eq:delta2}
\int_0^T   \delta(s) e^{ 2 \gamma \eta(s) \overline{B}} \ud s < \infty.
\end{gather}
Then, $\alpha$ is an admissible strategy, i.e. $\alpha  \in \mathcal A$.
\end{proposition}
The proof of Proposition \ref{ADM} is postponed to Appendix \ref{app:proofs}.

\section{The value function} \label{sec:valuef}

We define the value function corresponding to the optimization problem \eqref{pb:optimization} as
 \begin{align}\label{eq:valuefunction}
V(t,y^\a, y^\b, x):=\inf_{\alpha \in\mathcal{A}}{\mathbb{E}^{t,y^\a,y^\b,x}\bigl[ e^{-\gamma X_T^\alpha} \bigr]}, \quad (t,y^\a, y^\b, x) \in [0,T) \times \R^3,
\end{align}
where the notation $\mathbb{E}^{t,y^\a,y^\b,x}[\cdot]$ stands for the conditional expectation given $Y^\a_t=y^\a, Y^\b_t=y^\b$ and $X^\alpha_t=x$.

If $V$ is sufficiently smooth, by classical control theory results it can be characterized as the solution of the Hamilton-Jacobi-Bellman (in short HJB) equation
 \begin{equation}\label{eq:bellman}
\frac{\partial V}{\partial t}(t,y^\a, y^\b, x)+\inf_{\alpha \in [0,1]\times[0,1] \times \R}\mathcal{L}^{\alpha}_tV(t,y^\a, y^\b, x)=0, \quad   (t,y^\a, y^\b, x) \in [0,T) \times \R^3, \\
\end{equation}
with the final condition
\begin{equation}\label{eq:final_condition}
V(T,y^\a, y^\b, x) =  e^{-\gamma  x}, \quad  (y^\a, y^\b, x) \in \R^3,
\end{equation}
where, for any constant control $\alpha = (u^\a, u^\b, w) \in [0,1] \times  [0,1]  \times \R$, the operator $\mathcal L_t^\alpha$ denotes the infinitesimal generator of the process $(Y^\a, Y^\b, X^\alpha)$ which satisfies
\begin{align}
&\mathcal{L}^\alpha_tf(y^\a, y^\b, x)\\
& \ = \frac{\partial f}{\partial y^\a}b^\a(t)+\frac{\partial f}{\partial y^\b} b^\b(t)+ \frac{\partial f}{\partial x} \mu^X(t, x,y^\a,y^\b, \alpha)+\frac{1}{2}\frac{\partial^2 f}{\partial {y^\a}^2}({a^\a})^2(t)+\frac{1}{2}\frac{\partial^2 f}{\partial {y^\b}^2} ({a^\b})^2(t) \\
& \ +\frac{1}{2} \frac{\partial^2 f}{\partial x^2}w^2 \sigma^2(t)+
\int_0^{+\infty} \left(f(y^\a, y^\b, x - zu^\a) -f(y^\a, y^\b, x)\right) \lambda^\a(t, y^\a) F^\a(\ud z) \\
& \  + \int_0^{+\infty} \left(f(y^\a, y^\b, x - zu^\b - w K(t,z)) -f(y^\a, y^\b, x)\right) \lambda^\b(t,y^\b) F^\b(\ud z),
\end{align}
for every function $f$ which is $\mathcal C^2$ in $(y^\a, y^\b, x)\in \R^3$.

We will characterize the value function as the unique classical solution of the HJB equation\footnote{A classical solution is a function $v(t, y^\a, y^\b)$ that is $\mathcal C^1$ in $t$ and $\mathcal C^2$ in $(y^\b, y^\b, x)$ on $[0,T]\times \R^3$ and which satisfies equation \eqref{eq:bellman} with final condition \eqref{eq:final_condition}.} following a guess-and-very approach.
We consider the function
\begin{equation}\label{eq:ansatz}
V(t,y^\a, y^\b, x) = e^{ - \gamma x B(t,T)}  \psi (t,y^\a,y^\b),
\end{equation}
where $\psi$ is a function that does not depend on $x$. Plugging \eqref{eq:ansatz} into the HJB equation \eqref{eq:bellman} leads to the following reduced HJB equation
\begin{align}
&\ds \frac{\partial \psi}{\partial t}(t,y^\a, y^\b) + \mathcal{L}_t^{Y^\a,Y^\b}\psi(t,y^\a, y^\b) \\
& +  \psi(t,y^\a, y^\b) \Big ( \inf_{u^\a \in [0,1]} \Psi_1(t,y^\a,u^\b)
+ \inf_{u^\b \in [0,1],\ w \in \R} \Psi_2(t,y^\b,u^\b,w) \Big )=0,\label{eq:bellman1}
\end{align}
for all $(t, y^\a, y^\b) \in [0,T) \times \R^2$ with the terminal condition
\begin{equation}\psi(T,y^\a, y^\b) = 1, \quad (y^\a, y^\b) \in \R^2,\label{eq:final_condition1}\end{equation}
where $\mathcal{L}_t^{Y^\a,Y^\b}$ denotes the infinitesimal generator of the process $(Y^\a,Y^\b)$ given by
\begin{equation}
\mathcal{L}_t^{Y^\a,Y^\b}f(y^\a, y^\b)= \frac{\partial f}{\partial y^\a}b^\a(t)+\frac{\partial f}{\partial y^\b} b^\b(t)+\frac{1}{2}\frac{\partial^2 f}{\partial {y^\a}^2}({a^\a})^2(t)+\frac{1}{2}\frac{\partial^2 f}{\partial {y^\b}^2} ({a^\b})^2(t),
\end{equation}
and the functions $\Psi_1$ and $\Psi_2$ are respectively given by
\begin{align}
&\Psi_1(t,y^\a,u^\a)   \\
& \quad = \gamma B(t,T) \left(q^\a(t, y^\a, u^\a) - c^\a(t, y^\a)\right)  +  \int_0^{+\infty}  \biggl ( e^{\gamma  B(t,T)  zu^\a} - 1 \biggr ) \lambda^\a(t, y^\a) F^\a(\ud z)\label{psi1}\\
&\Psi_2(t,y^\b,u^\b,w)  \\
& \quad  = \gamma B(t,T) \biggl[\left(q^\b(t, y^\b, u^\b) - c^\b(t, y^\b)\right) + \frac{1}{2}\gamma B(t,T) \sigma^2(t)w^2  - w(\mu(t)  - r(t)) \biggr]\\
&\quad \quad  + \int_0^{+\infty}  \biggl ( e^{\gamma  B(t,T)(zu^\b + w K(t,z))}  - 1  \biggr ) \lambda^\b(t,y^\b) F^\b(\ud z).\label{psi2}
\end{align}

We start with a verification result.

\begin{theorem}[Verification Theorem]\label{Verification}
Let $\psi \in \mathcal{C}^{1,2,2}((0,T)\times \R^2) \cap \mathcal{C}([0,T]\times \R^2)$ be a classical solution of the HJB equation \eqref{eq:bellman1} with the final condition \eqref{eq:final_condition1}.
Then, the function $V(t,y^\a, y^\b, x)$ given in equation \eqref{eq:ansatz}
is the value function of \eqref{eq:valuefunction}.
Furthermore, let ${u^\a}^*(t,y^\a)$ and  $({u^\b}^*(t,y^\b), w^*(t,y^\b))$ be minimizers of  $\inf_{u^\a \in [0,1]} \Psi_1(t,y^\a,u^\b)$, and
 $\inf_{u^\b \in [0,1],\ w \in \R} \Psi_2(t,y^\b,u^\b,w) $, respectively. Then, $\alpha^*_t = ({u^\a}^*(t, Y^\a_t), {u^\b}^*(t,Y^\b_t), w^*(t,Y^\b_t)) \in\mathcal{A}$ is an optimal strategy.
\end{theorem}
The proof of the Verification Theorem is provided in Appendix \ref{app:proofs}.

Next, we look for a solution to the HJB equation \eqref{eq:bellman1} of the form
$\psi (t,y^\a,y^\b) = \psi_1(t,y^\a) \psi_2(t,y^\b)$. Plugging this function into \eqref{eq:bellman1} allows to split the problem in the following two backward PDEs
\begin{align}\label{eq:bellman2}
&\frac{\partial \psi_1}{\partial t}(t,y^\a)+ \mathcal{L}_t^{Y^\a}\psi_1 (t,y^\a) +  \psi_1 (t,y^\a)\inf_{u^\a \in [0,1]} \Psi_1(t,y^\a,u^\a) =0
\end{align}
\begin{align}\label{eq:bellman3}
\frac{\partial \psi_2}{\partial t}(t, y^\b)+ \mathcal{L}_t^{Y^\b}\psi_2 (t, y^\b) + \psi_2 (t, y^\b) \inf_{(u^\b, w) \in [0,1] \times \R} \Psi_2(t,y^\b,u^\b,w) =0,
\end{align}
with the final conditions $\psi_1(T,y^\a) = \psi_2(T,y^\b)= 1$, where $\mathcal{L}_t^{Y^\a}$ and $\mathcal{L}_t^{Y^\b}$ denote the infinitesimal generators of the processes $Y^\a$ ad $Y^\b$, respectively, and given by
\begin{align}
\mathcal{L}_t^{Y^\i}f(y^\i)= \frac{\partial f}{\partial y^\i}b^\i(t)+\frac{1}{2}\frac{\partial^2 f}{\partial {y^\i}^2}({a^\i})^2(t), \quad i=1,2.
\end{align}

A backward PDE similar to \eqref{eq:bellman2} is studied in \citet{brachetta2019proportional}  in  the case of constant risk-less interest rate. In \citet[Corollary 8.3]{brachetta2019proportional} sufficient conditions on $c^\a(t, y^\a)$, $q^\a(t, y^\a, u^\a)$, $\lambda^\a(t, y^\a)$ and $a^\a(t)$ are given to ensure existence and uniqueness of a classical solution. The result is reported, for completeness in Proposition \ref{P1} in Section \ref{sec:existence}.  

Instead, the backward PDE \eqref{eq:bellman3}, which accounts for dependence between the insurance and the financial markets, to the best of our knowledge, has not been previously discussed in the literature.

We now consider the minimization problems
\begin{equation}\label{min0}
 \inf_{u^\a \in [0,1]} \Psi_1(t,y^\a,u^\a) , \quad (t,y^\a) \in [0,T] \times  \R,
 \end{equation}
\begin{equation}\label{min}
 \inf_{(u^\b, w) \in [0,1] \times \R} \Psi_2(t,y^\b,u^\b,w) , \quad (t,y^\b) \in [0,T] \times  \R,
 \end{equation}
where the functions $ \Psi_1(t,y^\a,u^\a)$ and $\Psi_2(t,y^\b,u^\b,w)$ are given in \eqref{psi1} and \eqref{psi2}, respectively. We study problems \eqref{min0} and \eqref{min} and discuss their solutions in the next section.

\section{Solution of the minimization problems (\ref{min0}) and (\ref{min})}\label{sec:solution}

We begin with the problem \eqref{min0}.
\begin{proposition}\label{MC}
Assume that $ \frac{\partial^2  q^\a}{\partial {u^\a}^2} (t, y^\a, u^\a) \geq 0,$   for all $ (t, y^\a, u^\a) \in [0,T] \times \R \times [0,1]$.
Then, for any $(t,y^\a) \in [0,T] \times \R$, the function $\Psi_1(t,y^\a,u^\a)$ is strictly convex with respect to $u^\a \in [0,1 ]$ and there is a unique measurable function $u^{\a^{*}} (t,y^\a)$ given by
\begin{equation} \label{reins1}
u^{\a^{*}} (t,y^\a) = \begin{cases}
0  & (t,y^\a) \in A_0 \\
1 & (t,y^\a) \in A_1 \\
\bar u^\a(t,y^\a) & (t,y^\a) \in (A_0\cup A_1)^c
\end{cases}
\end{equation}
which solves the problem \eqref{min0}, where $(A_0\cup A_1)^c$ denotes the complementary of the set $A_0\cup A_1\subset [0,T]\times \R$, the sets $A_0$ and $A_1$ are given by
\begin{align}& A_0 = \left\{ (t, y^\a) \in [0,T] \times \R: \frac{\partial  q^\a}{\partial {u^\a}}(t, y^\a,0) + \lambda^\a(t, y^\a) \mathbb{E}[Z^\a]\ \geq 0\right\},\\
&A_1 = \left\{ (t, y^\a) \in [0,T] \times \R: \frac{\partial  q^\a}{\partial {u^\a}}(t, y^\a,1) + \lambda^\a(t, y^\a) \mathbb{E}[Z^\a
e^{ \gamma  B(t,T) Z^\a } ] \leq 0\right\},
 \end{align}
and  $\bar u^\a(t,y^\a) \in (0,1)$ is the unique solution of the equation
\begin{equation}
 \frac{\partial  q^\a}{\partial {u^\a}}(t,y^\a, u^\a) + \lambda^\a(t,y^\a) \int_0^{+\infty} z  e^{\gamma B(t,T) z u^\a } F^\a(\ud z) =0.
 \end{equation}
\end{proposition}

 \begin{proof}
We first observe that
\begin{align} &A_0 = \left\{ (t, y^\a) \in [0,T] \times \R:\ \frac{\partial  \Psi_1}{\partial {u^\a}} (t, y^\a,0) \geq 0\right\},\\
&A_1 = \left\{ (t, y^\a) \in [0,T] \times \R:\ \frac{\partial  \Psi_1}{\partial {u^\a}}(t, y^\a, 1) \leq 0\right\}.\end{align}
Then, the result follows by Proposition 4.1 and Lemma 4.1 in \citet{brachetta2019proportional}.
 \end{proof}

\begin{example} \label{x} We now discuss the optimal strategies for the reinsurance premia introduced in Example \ref{es:premi}.
\begin{itemize}
\item[(1)]
Under the {\em expected value principle}, for all $(t, y^\a, u^\a) \in [0,T] \times \R \times [0,1]$ we have that
 $$\frac{\partial q^\a}{\partial {u^\a}}(t, y^\a, u^\a) = - (1+\theta_R^\a)\lambda^\a(t, y^\a) \mathbb{E}[Z^\a].$$
In this case,  $ A_0=\emptyset$, which implies that full reinsurance is never optimal and $u^{\a^{*}} (t) = \min(1, \bar u^\a(t))$
with $\bar u^\a(t) \in (0, + \infty)$ being the unique solution of the equation
$$(1+\theta_R^\a) \esp{Z^\a} = \int_0^{+\infty} z  e^{\gamma B(t,T) zu^\a } F^\a(\ud z) =  \esp{Z^\a
e^{ \gamma  B(t,T) u^\a Z^\a } }.$$
\item[(2)]
Under the {\em variance principle}, for all $(t, y^\a, u^\a) \in [0,T] \times \R \times [0,1]$ we have that
 $$\frac{\partial q^\a}{\partial {u^\a}}(t, y^\a, u^\a) = - \lambda^\a(t, y^\a) \esp{Z^\a} - 2 \theta_R^\a \lambda^\a(t, y^\a) (1 - u^\a) \esp{(Z^\a)^2}.$$
Consequently, $ A_0= A_1= \emptyset$. In this case full and null reinsurances are never optimal and the optimal reinsurance $ u^{\a^{*}} (t) \in (0, 1)$ is the unique solution of the equation
\begin{equation}
\esp{Z^\a} + 2 \theta_R^\a(1- u^\a)  \esp{(Z^\a)^2} = \int_0^{+\infty} z  e^{\gamma B(t,T) z u^\a   } F^\a(\ud z).
\end{equation}
\end{itemize}
Note that in both cases the optimal reinsurance strategy is deterministic, that is,  it does not depend on the stochastic factor.
\end{example}

Now, we consider the  minimization problem \eqref{min}.
The next proposition shows that the function $\Psi_2$ is convex and hence a minimum exists. 

\begin{proposition}\label{convex}
Assume that $ \frac{\partial^2  q^\b}{\partial {u^\b}^2} (t, y^\b,u^\b)\geq 0,$  for all $(t, y^\b, u^\b) \in [0,T] \times \R \times [0,1]$.
Then, for all $(t,y^\b) \in [0,T] \times \R$, the function $\Psi_2(t,y^\b,u^\b,w)$ is  strictly convex in $(u^\b,w) \in [0,1 ] \times \R$.
\end{proposition}
\begin{proof}
To prove the result, we show that Hessian matrix of the function $\Psi_2$ is positive definite.
For any $(t,y^\b) \in [0,T] \times \R$, the second order derivatives of $\Psi_2$ are
\begin{align}
& \frac{\partial^2  \Psi_2}{\partial {u^\b}^2}(t,y^\b,u^\b,w)
=  \gamma B(t,T)  \frac{\partial^2  q^\b}{\partial {u^\b}^2}(t,y^\b,u^\b)\\
& \quad \quad +  \left(\gamma B(t,T) \right)^2 \lambda^\b(t, y^\b) \int_0^{+\infty} z^2 e^{\gamma B(t,T) ( zu^\b + w K(t,z)) } F^\b(\ud z), \label{eq:der3}
\end{align}
\begin{align}
& \frac{\partial^2  \Psi_2}{\partial  w \partial u^\b}(t,y^\b,u^\b,w) =  \left(\gamma B(t,T) \right)^2
\lambda^\b(t, y^\b) \int_0^{+\infty} z K(t,z)  e^{\gamma B(t,T) ( zu^\b + w K(t,z))} F^\b(\ud z), \label{eq:der4}
\end{align}
\begin{align}
& \frac{\partial^2  \Psi_2}{\partial  w^2}(t,y^\b,u^\b,w)=  \left(\gamma B(t,T) \right)^2
\biggl [  \sigma^2(t) +  \lambda^\b(t, y^\b) \int_0^{+\infty} K(t,z)^2 e^{\gamma B(t,T)( zu^\b + w K(t,z)) } F^\b(\ud z)   \biggr].\label{eq:der5}
 \end{align}
We note that \eqref{eq:der3} and \eqref{eq:der5} are positive, for each $(t,y^\b) \in [0,T] \times \R$.
Moreover, for every $(t,y^\b) \in [0,T] \times \R$,  by the Cauchy-Schwarz inequality, we have that
\begin{align}
&\left(\frac{\partial^2  \Psi_2}{\partial  w \partial u^\b} (t,y^\b,u^\b,w) \right)^2 \leq
\left(\gamma B(t,T) \right)^4 (\lambda^\b(t, y^\b))^2 \\
& \quad \times \int_0^{+\infty} z^2 e^{\gamma B(t,T)( zu^\b + w K(t,z))} F^\b(\ud z) 
\int_0^{+\infty} K(t,z)^2 e^{\gamma B(t,T)( zu^\b + w K(t,z))} F^\b(\ud z)\\
& \quad <  \frac{\partial^2  \Psi_2}{\partial  {u^\b}^2} (t,y^\b,u^\b,w) \ \frac{\partial^2  \Psi_2}{\partial  w^2} (t,y^\b,u^\b,w),
\end{align}
which implies that the determinant of the Hessian matrix of the function $\Psi_2$ is positive definite for all $(u^\b,w) \in [0,1 ] \times \R$.
\end{proof}

We now consider the function $q^\b(t,y^\b, u^\b)$, with the same properties of Definition \eqref{def:reinsurance_premium}, extended for $u^\b <0$ and $u^\b >1$. 
For all $(t, y^\b,  u^\b,w) \in [0,T] \times \R^3$ we introduce the functions
\begin{align}\label{H}
H(t,y^\b, u^\b,w) &= \lambda^\b(t,y^\b) \int_0^{+\infty} z  e^{\gamma B(t,T)( zu^\b  + w K(t,z)) } F^\b(\ud z) + \frac{\partial  q^\b}{\partial {u^\b}}(t,y^\b, u^\b),\\
\widetilde H (t,y^\b, u^\b,w)&= \gamma B(t,T) \sigma^2(t) w -\! (\mu(t) - r(t) )\\
 \label{eq:Htilde} &+ \lambda^\b(t, y^\b) \int_0^{+\infty} K(t,z) e^{\gamma B(t,T)( zu^\b + w K(t,z)) } F^\b(\ud z).
\end{align}

The following result is an auxiliary lemma which will be applied in the proof of Proposition \ref{existence} below, where the candidate optimal reinsurance strategy corresponding to the second line of business is characterized. 

\begin{lemma}\label{lemma:minimizer}
Assume that $\frac{\partial^2  q^\b}{\partial {u^\b}^2} (t, y^\b, u^\b)\geq 0,$ $(t, y^\b,  u^\b) \in [0,T] \times \R^2$ and that
 \begin{equation}\label{H1}
\lim_{u^\b \to  - \infty} H(t,y^\b, u^\b, w)  <0, \quad  \lim_{u^\b \to  + \infty} H(t,y^\b, u^\b, w) >0,
\end{equation}
for all $(t, w, y^\b) \in [0,T] \times \R^2.$
Then, for all $(t, y^\b) \in [0,T] \times \R$, the system
\begin{align}
&H(t,y^\b, u^\b, w) =0    \label{deriv}\\
&\widetilde H(t,y^\b, u^\b, w) =0, \label{derivgen}
 \end{align}
has a unique solution $(\bar u^{\b}(t, y^\b), \bar w(t, y^\b)) \in \R^2$.
\end{lemma}

The proof of Lemma \ref{lemma:minimizer} is given in Appendix \ref{app:proofs}.
Here, we stress one of the properties of the solution which will be also used later. The solution $\widetilde w(t,y^\b, u^\b)$ of equation \eqref{derivgen} is monotonic decreasing in $u^\b\in \R$, for every  $(t, y^\b) \in  [0,T] \times \R$ (see Step 2 in the proof of Lemma \ref{lemma:minimizer}).  This implies that,
for every $u^\b \in (0,1)$, $\widetilde w (t, y^\b, u^\b)$ is bounded from above and below by the solutions of equation \eqref{derivgen} with $u^\b=0,1$. Explicitly,
\begin{equation}\label{monotonia}
\widetilde w(t, y^\b, 1) <  \widetilde w(t, y^\b, u^\b) <  \widetilde w(t, y^\b, 0),
\end{equation}
for all $u^\b\in (0,1)$. For the sake of clarity, in Figure \ref{f1} we plot $\widetilde H(t,y^\b,u^\b,w)$ as a function of $w$ for fixed $(t,y^\b)$ corresponding to three different values of $u^\b$, namely,  $u^\b=0,1,\frac{1}{2}$. In this example we choose $t=0$, $y^\b=-0.2$, $\lambda^\b(0,-0.2)=10e^{0.2}$, $\gamma=0.5$, $r(0)=0.02$, $\mu(0)=0.05$, $\sigma(0)=0.1$, and the jump size distribution $F^\b(z)$ is assumed to be exponential with expectation equal to $1$. The function $\widetilde H(t,y^\b,u^\b,w)$ intercepts the horizontal axis at values $\widetilde w(t, y^\b, 1),\widetilde w(t, y^\b, 0.5),\widetilde w(t, y^\b, 0)$, which satisfy the chain of inequalities \eqref{monotonia}.
\begin{figure}[h]
  \centering
  \includegraphics[width=.6\textwidth]{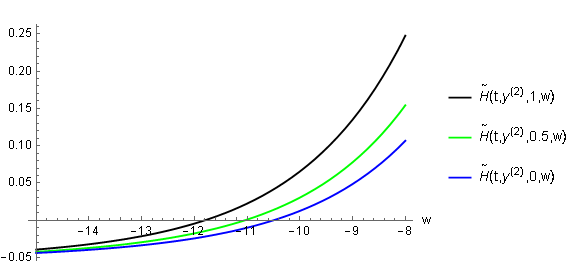}
  \caption{The function $\widetilde H(t,y^\b,u^\b,w)$ as a function of $w$ for fixed $(t,y^\b)$ and values of $u^\b=0,1,0.5$. The function $\widetilde H(t,y^\b,u^\b,w)$ intercepts the horizontal axis at values $\widetilde w(t, y^\b, 1),\widetilde w(t, y^\b, 0.5),\widetilde w(t, y^\b, 0)$ which satisfy the chain of inequalities \eqref{monotonia}.}\label{f1}
\end{figure}

\begin{proposition}\label{existence}
Under the same hypotheses of Lemma \ref{lemma:minimizer}, there is a unique minimizer $(u^{\b^*}(t,y^\b), w^*(t,y^\b))$ of problem \eqref{min} given by
\begin{equation} \label{reinsgen}
u^{\b^*} (t,y^\b) = \max \left \{ \min \left \{ \bar u^{\b}(t,y^\b), 1 \right \},0 \right \},
\end{equation}
for all $(t,y^\b)  \in  [0,T] \times \R,$ and
 \begin{equation} \label{reins2}
w^*(t,y^\b)= \begin{cases}
\widetilde w(t,y^\b,0)  & (t,y^\b) \in \widetilde A_0 \\
\widetilde w(t,y^\b,1) & (t,y^\b) \in \widetilde A_1 \\
\bar w(t,y^\b) & (t,y^\b) \in (\widetilde A_0\cup \widetilde A_1)^c,
\end{cases}
\end{equation}
where \begin{itemize}
\item[(i)] $(\bar u^{\b}(t, y^\b), \bar w(t, y^\b))$ is the unique solution of the system of equations \eqref{deriv}--\eqref{derivgen};
\item[(ii)]
 $\widetilde w(t,y^\b,0)$,  and $\widetilde w(t,y^\b,1)$ are unique solutions of equation \eqref{derivgen} with $u^\b=0,1$, respectively (see Step 2 in the proof of Lemma \ref{lemma:minimizer});
\item[(iii)] $(\widetilde A_0\cup \widetilde A_1)^c$ is the complementary of the set $\widetilde A_0\cup \widetilde A_1$ and the sets $\widetilde A_0$ and $\widetilde A_1$ are given by
    \begin{align}\widetilde A_0:=\{ (t,y^\b)  \in  [0,T] \times \R:\ \bar u^\b(t,y^\b) \leq 0 \},\\
    \widetilde A_1:=\{ (t,y^\b)  \in  [0,T] \times \R:\ \bar u^\b(t,y^\b) \geq 1 \}.
    \end{align}
    \end{itemize}
\end{proposition}

\begin{proof}
Let us observe that the first order conditions in problem \eqref{min} read as the system \eqref{deriv}--\eqref{derivgen}.
Then, by Proposition \ref{convex} and by Lemma \ref{lemma:minimizer} , the minimizer of the problem \eqref{min} on the set
$$\left\{ (t,y^\b)  \in  [0,T] \times \R:  0<\bar u^{\b}(t, y^\b) < 1 \right\}$$ is given by $(u^{\b^{* }} (t,y^\b) = \bar u^\b(t, y^\b),  w^*(t,y^\b) = \bar w(t, y^\b))$.



The following three cases arise.
\begin{itemize}
\item[(i)]If $(t,y^\b) \in \widetilde A_0$, using the fact that $\widetilde w(t, y^\b, u^\b)$ is monotonic in $u^\b$ (see the inequalities  in \eqref{monotonia}) and that $\frac{\partial  \Psi_2}{\partial u^\b}(t,y^\b,u^\b,w)$ is increasing in $u^\b \in [0,1]$ and in $w \in \R$, it holds that
$$\frac{\partial  \Psi_2}{\partial u^\b}(t, y^\b, 0, \widetilde w(t,y^\b,0))  \geq \frac{\partial  \Psi_2}{\partial u^\b}  (t, y^\b, \bar u^\b(t, y^\b) , \bar w(t, y^\b))=0,$$
because $0 > \bar u^\b(t, y^\b)$ and $\widetilde w(t,y^\b,0) > \bar w(t, y^\b)$. This implies that
$\Psi_2(t, y^\b, u^\b, \widetilde w(t,y^\b,0))$  is increasing in $u^\b \in [0,1]$ and hence, for $(t,y^\b) \in \widetilde A_0$, the minimizer of problem \eqref{min} is given by $(u^{{\b}^{* }} (t,y^\b) = 0, w^*(t,y^\b) = \widetilde w(t,y^\b,0))$.
\item[(ii)]
If $(t,y^\b) \in \widetilde A_1$, since $\widetilde w(t, y^\b, u^\b)$ is monotonic in $u^\b$ (see \eqref{monotonia}) and  $\frac{\partial  \Psi_2}{\partial u^\b}(t,y^\b,u^\b,w)$ is increasing in $u^\b \in [0,1]$ and in $w \in \R$, we also get that
$$\frac{\partial  \Psi_2}{\partial u^\b}(t, y^\b, 1, \widetilde w(t,y^\b,1))  \leq \frac{\partial  \Psi_2}{\partial u^\b}  (t, y^\b, \bar u^{\b}(t, y^\b) , \bar w(t, y^\b))=0.$$
Then,
$\Psi_2(t, y^\b, u^\b, \widetilde w(t,y^\b,1))$  is decreasing in $u^\b \in [0,1]$, and the minimizer of problem \eqref{min} for $(t,y^\b) \in \widetilde A_1$, is given by $(u^{{\b}^{* }} (t,y^\b) = 1, w^*(t,y^\b) = \widetilde w(t,y^\b,1))$.
\item[(iii)]
Finally, if $(t,y^\b) \in (A_0\cup A_1)^c$, the minimizer of problem \eqref{min} is given by $(u^{{\b}^{*}} (t,y^\b) = \bar u^{\b}(t, y^\b),  w^*(t,y^\b) = \bar w(t, y^\b))$.
\end{itemize}
\end{proof}

 \begin{remark}
Note that the assumptions of Lemma \ref{lemma:minimizer} (and hence of Proposition \ref{existence}) are fulfilled under the classical premium principles. Indeed, under
the expected value principle we have
 \begin{equation}
H(t,w, y^\b, u^\b)=   \lambda^\b(t,y^\b) \left(\int_0^{+\infty} z  e^{\gamma B(t,T)( zu^\b  + w K(t,z)) } F^\b(\ud z) - (1+\theta_R^\b) \mathbb{E} [Z^\b]\right) \end{equation}
and under the variance principle
 \begin{equation}
 \begin{split}
& H(t,w, y^\b, u^\b)\\
& \qquad =   \lambda^\b(t,y^\b)\left( \int_0^{+\infty} z  e^{\gamma B(t,T)( zu^\b  + w K(t,z)) } F^\b(\ud z) -  \mathbb{E} [Z^\b] - 2 \theta_R^\b(1- u^\b)  \mathbb{E} [(Z^\b)^2]\right).
\end{split}
\end{equation}
\end{remark}


\begin{proposition}[Properties of $w^{*}$]\label{stime}
Suppose that the hypotheses of Lemma \ref{lemma:minimizer} hold.  Then,
\begin{equation}\label{dis_1}
w^*(t,y^\b) \leq \frac{\mu(t) - r(t)}{\gamma B(t,T) \sigma^2(t)},
\end{equation}
for all $(t,y^\b)  \in  [0,T] \times \R$. 
In particular, if $\mathbb{E}[K(t,Z^\b)]=\int_0^{+\infty} K(t,z)F^\b(\ud z) >0$,  for all $t \in [0,T]$, inequality \eqref{dis_1} is strict. Moreover, 
\begin{equation}\label{dis_2}
w^*(t,y^\b) \geq \min \left\{0, \frac{\mu(t) - r(t)}{\gamma B(t,T) \sigma^2(t)} - \frac{  \delta(t) \int_0^{+\infty} e^{\gamma B(t,T) z } F^\b(\ud z)}{\gamma B(t,T) \sigma^2(t)} \right\},
\end{equation}
for all $(t,y^\b)   \in  [0,T] \times \R$ and where $\delta(t)$ is introduced in Assumption \ref{ASS}.
Finally, if we set
\begin{align} \label{c1}
C_1 &:= \left\{(t, y^\b) \in  [0,T]\times \R:\   \mu(t) - r(t)  <  \lambda^\b(t, y^\b) \int_0^{+\infty}  K(t,z) F^\b(\ud z)\right\}, \\
\label{c2} C_2 &:= \left\{ (t, y^\b) \in  [0,T]\times \R:\  \mu(t) - r(t)  > \lambda^\b(t, y^\b) \int_0^{+\infty} K(t,z) e^{\gamma B(t,T) z } F^\b(\ud z)\right\}, \qquad{}  \end{align} then we get that $w^*(t,y^\b) < 0$, for all $(t,y^\b)   \in C_1$ and $w^*(t,y^\b) > 0$, for all $(t,y^\b)   \in C_2$.
\end{proposition}
The proof of Proposition \ref{stime} can be found in Appendix \ref{app:proofs}. Notice that the case $\mathbb{E}[K(t, Z^\b)]>0$, for all $t \in [0,T]$, corresponds to the case where common shock is not negligible.

\subsection{Expected value principle}\label{EVP}

Now, we discuss the case where the insurance and reinsurance premia  $c^\b(t, y^\b), q^\b(t, y^\b, u^\b)$  are computed according to the expected value principle, that is,
\begin{align}
 c^\b(t, y^\b)  & =  (1+\theta^\b)\mathbb{E}[Z^\b]\lambda^\b(t,y^\b), \\
 q^\b(t, y^\b, u^\b)  & = (1+\theta_R^\b)\mathbb{E} [Z^\b] (1-u^\b) \lambda^\b(t,y^\b).
\end{align}
We derive explicit formulas for the optimal strategy under the following additional assumptions on the distribution of the claim size and the jumps of the price process.
\begin{assumption}\label{moltiplicative}
The random variable $Z^\b$ takes values in a compact set and
$K(t,z) = k (t)  z$, with $0<k (t) < \frac{1}{D}$, where $D:= \inf\{ z \in (0, +\infty): \  F^\b(z) = 1\} <  \infty$.
\end{assumption}

Equation \eqref{psi2}, in this case,  reads as
\begin{align}\label{psi3}
 \Psi_2(t,y^\b,u^\b,w)  & = \gamma B(t,T) \mathbb{E}[Z^\b]\lambda^\b(t,y^\b) \left[ \theta_R^\b -  \theta^\b - ( 1 + \theta_R^\b) u^\b\right] \\
&  + \gamma B(t,T)
 \biggl [  \frac{1}{2} \gamma B(t,T) \sigma^2(t) w^2 - (\mu(t) - r(t) ) w\biggr] \\
&  +
\lambda^\b(t,y^\b) \int_0^{+\infty} \left( e^{\gamma B(t,T) z (u^\b + w k(t))}  - 1  \right) F^\b(\ud z).
\end{align}

Now, we consider the first order conditions
\begin{align}
H(t, u^\b, w)=0, \label{derivata1}\\
\widetilde H (t, y^\b, u^\b, w)=0, \label{derivata2}
\end{align}
with
 \begin{align}
H(t,  u^\b, w)=&- \lambda^\b(t,y^\b)(1+\theta_R^\b) \mathbb{E} [Z^\b] +  \lambda^\b(t,y^\b)\int_0^{+\infty} z  e^{\gamma B(t,T) z (u^\b  + w k(t)) } F^\b(\ud z), \\
\widetilde  H (t, y^\b, u^\b, w)=&\gamma B(t,T) \sigma^2(t) w - (\mu(t) - r(t) )\\
&+ k(t) \lambda^\b(t, y^\b)  \int_0^{+\infty}  z e^{\gamma B(t,T) z (u^\b + w k(t)) } F^\b(\ud z).
 \end{align}
We set  $\phi (t)  = u^\b  + w k(t)$. Then,  \eqref{derivata1} is equivalent to
 \begin{equation} \label{der1}(1+\theta_R^\b) \mathbb{E} [Z^\b] = \int_0^{+\infty} z  e^{\gamma B(t,T) z \phi} F^\b(\ud z).
 \end{equation}
Equation \eqref{der1} has a unique deterministic solution $\phi^*(t)>0$. Indeed, define the function
 $$
 h(t,\phi)= \int_0^{+\infty} z  e^{\gamma B(t,T) z \phi } F^\b(\ud z);
 $$
for all $t \in [0,T]$, we get $h(t,0) = \mathbb{E} [Z^\b] < (1+\theta_R^\b) \mathbb{E} [Z^\b]$, and also,  $h(t,\phi)$ is continuous and increasing in $\phi$ and satisfies $\lim_{\phi \to  + \infty} h(t,\phi) = + \infty$. Figure \ref{f2} provides a representation of the function $\phi \mapsto h(t,\phi)$ at $t=0$ under the parameters $r(0)=0.02, \gamma=\frac{1}{2}$ and $\theta_R^\b=0.3$. Here, the jump size distribution $F^\b(z)$ is assumed to be truncated exponential with the density $f^\b(z)=c e^{-z} \I_{[0,100]}(z)$ and $c=\frac{1}{1-e^{-100}}$. We see that the function intercepts the level $(1+\theta_R^\b) \mathbb{E} [Z^\b]$ at a unique point.
\begin{figure}[h]
  \centering
  \includegraphics[width=.6\textwidth]{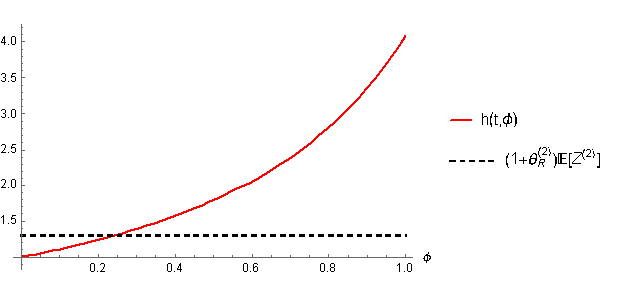}
  \caption{The function $h(t,\phi)$ as a function of $\phi$ at $t=0$. The function $\widetilde H(t,y^\b,u^\b,w)$ intercepts the level $(1+\theta_R^\b) \mathbb{E} [Z^\b]$ at a unique point $\phi^*$.}\label{f2}
\end{figure}

Now let
\begin{equation} \label{u bar}
\bar u^\b(t,y^\b) = \phi^*(t)  - \bar w(t,y^\b) k(t),
\end{equation}
where $\phi^*(t)>0$ is the unique  solution of the equation \eqref{der1}
 and
\begin{equation} \label{inv}  \bar w(t,y^\b) = \frac{\mu(t) - r(t)}{\gamma B(t,T) \sigma^2(t)} - \frac{\lambda^\b(t, y^\b) k(t) (1+\theta_R^\b) \mathbb{E} [Z^\b]} {\gamma B(t,T) \sigma^2(t)},
\end{equation}
for all $(t,y^\b)  \in  [0,T] \times \R$.
Similarly to Proposition \ref{existence}, we define  the sets
\begin{align}
\widetilde A_0 & =\{ (t,y^\b)  \in  [0,T] \times \R:\ \bar u^\b(t,y^\b) \leq 0 \}, \\
\widetilde A_1 & =\{ (t,y^\b)  \in  [0,T] \times \R:\ \bar u^\b(t,y^\b) \geq 1 \},\\
(\widetilde A_0\cup \widetilde A_1)^c & = \{ (t,y^\b)  \in  [0,T] \times \R:\  0 <\bar u^\b(t,y^\b) <1 \},
\end{align}
and let the functions $\widetilde w(t,y^\b,0)$ and $\widetilde w(t,y^\b,0)$ be the unique solutions of equation \eqref{derivata2} with $u^\b=0,1$, respectively.

 \begin{corollary}\label{special}
Under the expected premium principle, problem \eqref{min} admits a unique minimizer $(u^{\b^*}(t,y^\b), w^*(t,y^\b))$
given by
\begin{equation} \label{reins}
u^{{\b}^{* }} (t,y^\b) = \max \left \{ \min \left \{ \bar u^\b(t,y^\b), 1 \right \},0 \right \},
\end{equation}
for all $(t,y^\b)  \in  [0,T] \times \R$, and
\begin{equation} 
w^*(t,y^\b)= \begin{cases}
\widetilde w(t,y^\b,0)  & (t,y^\b) \in \widetilde A_0, \\
\widetilde w(t,y^\b,1) & (t,y^\b) \in \widetilde A_1, \\
\bar w(t,y^\b) & (t,y^\b) \in (\widetilde A_0\cup \widetilde A_1)^c.
\end{cases}
\end{equation}
\end{corollary}

\begin{proof}
Using the definition \eqref{der1} of $\phi^*(t)$ and plugging it into equation \eqref{derivata2}
leads to
\begin{equation}\label{eq:der22}\gamma B(t,T) \sigma^2(t) w - (\mu(t) - r(t) )  = - \lambda^\b(t, y^\b) k(t) (1+\theta_R^\b) \mathbb{E} [Z^\b],\end{equation}
for all $(t,y^\b)  \in  [0,T] \times \R$. Equation \eqref{eq:der22} has a unique solution given by
\begin{equation} \label{inv}  \bar w (t,y^\b) = \frac{\mu(t) - r(t)}{\gamma B(t,T) \sigma^2(t)} - \frac{\lambda^\b(t, y^\b) k(t) (1+\theta_R^\b) \mathbb{E} [Z^\b]} {\gamma B(t,T) \sigma^2(t)},
\end{equation}
for all $(t,y^\b)  \in  [0,T] \times \R$. The rest of the proof follows by Proposition \ref{existence}.
\end{proof}

\section{The optimal reinsurance-investment strategy and its properties}\label{sec:strategy}

We now collect the results of the previous sections and derive the optimal reinsurance and investment strategy corresponding to the problem \eqref{pb:optimization}.
\begin{proposition}\label{ultima}
Under the hypotheses of Lemma \ref{lemma:minimizer}, assume existence and uniqueness of a classical solution to the HJB equation \eqref{eq:bellman1} with the final condition \eqref{eq:final_condition1}. Moreover, suppose that for all $t \in [0,T]$
\begin{align}\label{eq:sigma}\sigma(t) > \sigma>0, \end{align}
and that \begin{align}\label{eq:kappa}
\int_0^T \delta(t) e ^{\kappa \delta(t) } \ud t < \infty,
\end{align}
where $\kappa = \frac{ 2 }  { \sigma^2} \overline{B}  \mathbb{E} \left[ e^{\gamma  \overline{B}  Z^\b } \right] $, $\overline{B}$ is defined in Notation \ref{notation:B} and $\delta(t)$ is introduced in Assumption \ref{ASS}.
Then, the process  $\{(u^{\a^*}(t,Y^\a_t), u^{\b^*} (t, Y^\b_t), w^*(t,Y^\b_t)): t \in [0,T]\}$,
where  $u^{\a^*}(t, y^\a)$ is computed in Proposition \ref{MC} and the pair $(u^{\b^*}(t,y^\b), w^*(t,y^\b))$ is computed in Proposition \ref{existence},
provides the optimal strategy for the problem \eqref{pb:optimization}.
\end{proposition}

\begin{proof}
Let $$\eta(t) = \max \left\{ \frac{|\mu(t) - r(t)|}{\gamma B(t,T) \sigma^2(t)},  \frac{\delta(t)}{ \gamma B(t,T) \sigma^2(t)} \int_0^{+\infty} e^{\gamma B(t,T) z } F^\b(\ud z)  \right\}, \quad  t \in [0,T].$$
We show that conditions \eqref{w_ass}, \eqref{eq:delta1} and \eqref{eq:delta2} in Proposition \ref{ADM} are satisfied.

By Proposition \ref{stime}, we immediately have that $|w^*(t,Y^\b_t)| \leq \eta(t)$, for every $t \in [0,T]$.
Moreover, by \eqref{eq:sigma} and Assumption \ref{ASS}, we get
\begin{align}
&\int_0^T \Big(\eta(s)( |\mu(s) - r(s)| +\delta(s))+ {\eta}^2(s) \sigma^2(s) \Big) \ud s < \infty.
\end{align}

Finally, by \eqref{eq:kappa}
\begin{align}
 &\int_0^T  \delta(t) e^{ \frac{2 \gamma \overline{B}}{\gamma B(t,T)\sigma^2(t)} \delta(t) \int_0^{+\infty} e^{\gamma B(t,T) z } F^\b(\ud z)} \ud t  \leq \int_0^T  \delta(t) e^{ \frac{2 \overline{B}} {\sigma^2(t)} \delta(t) \mathbb{E}[ e^{\gamma  \overline{B}  Z^\b }]} \ud t =\int_0^T \delta(t) e ^{\kappa \delta(t)} dt <  \infty,
\end{align}
which implies \eqref{eq:delta2}.
Then, the triplet $\{(u^{\a^*}(t,Y^\a_t), u^{\b^* } (t, Y^\b_t), w^*(t,Y^\b_t)),\ t \in [0,T]\}$ is an admissible strategy and the statement follows by applying the Verification Theorem \ref{Verification} and Proposition \ref{existence}.
\end{proof}

\begin{remark}
If the arrival intensity of catastrophic claims is assumed to be bounded, i.e. for all $t \in [0,T] \times \mathbb R$, $\lambda^\b(t, y^\b) \leq \delta$, for a positive constant $\delta$, then condition  \eqref{eq:kappa} is trivially satisfied.
\end{remark}

\subsection{The effect of common shocks}\label{sec:effects}

To underline the characteristics of our model we compare optimal investment and reinsurance strategies with common shock dependence between  the insurance risk process and the risky asset, with those corresponding to a modelling framework where common shocks is not accounted for.

If there is no common shock, that is, when $K(t,z)=0$, the optimal investment and reinsurance strategy is given by  $\{(u^{{\a}^*}(t,Y^\a_t),  u^{{\b,no}^*} (t, Y^\b_t),  w^{{no}^*}(t)),\ t \in [0,T]\}$, where 
 $$ w^{{no}^*}(t) = \frac{\mu(t) - r(t)}{\gamma B(t,T) \sigma^2(t)}, \; u^{{\b,no}^*} (t,y^\b) = \max\{\min \{\bar u^{\b,no}(t,y^\b),1\}, 0\},$$
and $ \bar u^{\b,no}(t,y^\b)$ is the unique solution of the equation
\begin{equation}\label{nocomm}
 H^{no}(t,y^\b, u^\b) = \frac{\partial  q^\b}{\partial {u^\b}}(t,y^\b, u^\b) + \lambda^\b(t,y^\b) \int_0^{+\infty} z  e^{\gamma B(t,T)z u^\b} F^\b(\ud z) =0.
 \end{equation}
When there is no common shock, the optimal investment strategy turns out to be deterministic (i.e. function of time only) and the reinsurance strategies $\{(u^{{\a}^*}(t,Y^\a_t),  u^{{\b,no}^*} (t, Y^\b_t),\ t \in [0,T]\}$ depend in general on the stochastic factors $Y^\a$, $Y^\b$ that affect claim arrival intensities of each business line.
Under classical premium principles, the dependence of reinsurance strategies on the factors $Y^\a$, $Y^\b$ disappears, as seen in Example \ref{x}.

Below, we provide a comparison result between the common shock setting and the no-common shock setting.

\begin{proposition}\label{comparison}
Assume that $\mathbb{E}[K(t,Z^\b)] >0$, for all $t \in [0,T]$ (i.e. common shock is not negligible), the hypotheses of Proposition \ref{ultima} hold, and let $C_1$ and $C_2$ be defined in \eqref{c1} and \eqref{c2}, respectively. Then, the following inequalities hold
\begin{itemize}
\item[(i)] $w^*(t, Y^\b_t) < w^{*,no}(t)$,  for all $t \in [0,T]$;
\item[(ii)] $u^{{\b}^{*} } (t,Y_t^\b) >  u^{{\b,no}^*} (t,Y_t^\b)$ for all $(t,Y^\b) \in C_1$;
\item[(iii)]  $u^{\b^*} (t,Y_t^\b) < u^{{\b,no}^*} (t,Y_t^\b)$, for all $(t,Y^\b) \in C_2$.
\end{itemize}
Moreover, for all $\in [0,T]$, the values of $w^*(t, Y^\b_t)$ get smaller as the intensity $\lambda^\b(t,Y^\b_t)$ of claim arrival as well as with the jump size of the risky asset due to the shock (i.e. $K(t,Z^\b)$) get larger.
 \end{proposition}

\begin{proof}
The property (i) follows directly by Proposition \ref{stime}. We now show conditions (ii) and (iii).
For all $w >0$, we have that $H(t,y^\b, u^\b,w) > H^{no}(t,y^\b, u^\b)$, for all $(t, y^\b, u^\b) \in [0,T] \times \mathbb R^2$, and for all $ w <0$ the inequality reverses, $H(t,y^\b, u^\b,w) < H^{no}(t,y^\b, u^\b)$, for all $(t, y^\b, u^\b) \in [0,T] \times \mathbb R^2$.
These inequalities imply  that, if the optimal quote invested in the risky asset is positive, then the optimal retention level under common shock is lower than the case where no-common shock, that is  $u^{{\b}^{*} } (t,Y^\b_t) <  u^{{\b,no}^*} (t,Y^\b_t)$. Conversely, if the optimal quote invested in the risky asset is negative we have that
$u^{\b^*} (t,Y_t^\b) > u^{{\b,no}^*} (t,Y_t^\b)$. Hence, properties (ii) and (iii) follow by applying Proposition \ref{stime}.

Finally, the monotonic dependence of the optimal investment strategy with common shock by the quantities $\lambda^\b(t,Y^\b_t)$ and $K(t,Z^\b)$  follows by the form of the function $\widetilde H$ given in \eqref{eq:Htilde}.
\end{proof}

We now make a few observations on the behavior of the insurance company.

\begin{itemize}
\item[1.]  The optimal quote invested in the risky asset in case of common shock dependence, relies upon the stochastic
factor $Y^\b$ which affects the arrival intensity of the second line of business. Moreover, this is always smaller than the optimal investment strategy without common shock.
Instead, the optimal reinsurance strategy in case of common shock dependence does not always dominate the case of no-common shock. Whether it is smaller or larger than the optimal reinsurance without common shock depends on the optimal investment strategy.\\
If the impact of the shock on the risky asset price is small enough, i.e. if $(t,Y^\b_t) \in C_2$,  the insurance company invests a positive quote in the risky asset (i.e. buys the risky asset) and  the optimal proportional quote to be transferred to reinsurance is larger than in case of no-common shock. Conversely, if the impact of the shock on the risky asset price is large, i.e. if $(t,Y^\b_t) \in C_1$, the insurance company has a completely different behavior: it short-sells the asset and buys less reinsurance. In other words, the insurance company transfers to the reinsurance a smaller percentage of its claims compared to the case without common shock.
\item[2.]Due to the form of the function $\widetilde H$ in equation \eqref{eq:Htilde}, we get that for all $t \in [0,T]$,
the difference  between the investment strategies without and with common shock, $w^{*,no}(t) - w^*(t, Y^\b_t)$ increases when the intensity $\lambda^\b(t,Y^\b_t)$ of claim arrival as well as with the jump size of the risky asset due to the shock, that is the quantity $K(t,Z^\b)$,  increase.
\item[3.]
An analogous effect of common shocks on the optimal strategies is obtained by \citet{liang2018optimal} under different modelling assumptions, i.e. when arrival intensities are assumed to be constant and  risky asset jump size  (induced by the common shock) are taken independent of  the claim sizes. 
In \citet{liang2018optimal}, it is shown under the classical premium principle that 
whether the optimal reinsurance strategy with common shock is smaller or larger than that without common shock depends on the values of the parameters of the model and that the optimal investment strategy with common shock is always smaller than the one without common shock dependence. Moreover, it decreases as the claims arrival intensity parameter gets larger.
\end{itemize}

\subsection*{Comparison under the Expected Value Principle}
Under the expected value principle discussed in Section \ref{EVP}, the case of no-common shock is obtained by taking $k(t)=0$. The optimal  strategy is given
by $\{(u^{{\a}^*}(t), u^{{\b,no}^{* }} (t), w^{*,no}(t)),\ t \in [0,T]\}$,
where
$u^{{\b,no}^{* }} (t) = \min \{1, \bar u^{\b,no}(t)\}$, $ \bar u^{\b,no}(t)>0$ is the unique solution to the following equation
$$(1+\theta_R^\b) \mathbb{E} [Z^\b] = \int_0^{+\infty} z  e^{\gamma B(t,T) z u^\b} F^\b(\ud z),$$
and $$ w^{*,no}(t) = \frac{\mu(t) - r(t)}{\gamma B(t,T) \sigma^2(t)}.$$ Hence, the optimal  strategy is a function of time only.
Accounting for common shock introduces an additional stochastic term in the investment strategy given by $- \frac{\lambda^\b(t, Y^\b_t) k(t) (1+\theta_R^\b) \mathbb{E} [Z^\b]} {\gamma B(t,T) \sigma^2(t)}$, see equation \eqref{inv}.
In this specific case, we can detect explicitly the effect of common shock on the optimal reinsurance strategy through the quantity
$$\bar u^{\b}(t, Y^\b_t) = \bar u^{\b,no}(t) - k(t) \bar w(t, Y^\b_t).$$
%
%
The same considerations of the general case apply to this example, too. In particular, since for all $t \in [0,T]$, $\bar w(t, Y^\b_t)$ gets smaller as  $ \lambda^\b(t, Y^\b_t)$ and $k(t)$ gets larger (see equation \eqref{inv}) and $\bar u^{\b,no}(t)$ does not depend on $ \lambda^\b(t, Y^\b_t)$ and $k(t)$, then  $\bar u^{\b}(t, Y^\b_t)$ increases with $ \lambda^\b(t, Y^\b_t)$ and $k(t)$ for all $t \in [0,T]$.
Consequently, the optimal reinsurance strategy in the second line of business $u^{\b^*}(t, Y^\b_t)$, for each $t \in [0,T]$,  increases whenever the claims arrival intensity $ \lambda^\b(t, Y^\b_t)$ and amplitude of risky asset jump $k(t)$ increase.

Moreover, the optimal reinsurance strategy in the second line of business $u^{\b^*}(t, Y^\b_t)$ increases with respect to the reinsurance safety loading $\theta_R^\b$, for all $t \in [0,T]$. This comes as a consequence of the fact that $\bar w(t, Y^\b_t)$ is decreasing and $\bar u^{\b,no}(t)$ is increasing, with respect to the reinsurance safety loading $\theta_R^\b$ for all $t \in [0, T]$.
The fact that reinsurance premium increases with the claim arrival intensity and the reinsurance safety loading explains monotonic property of the optimal reinsurance strategy with respect to these quantities. The effect produced by amplitude of risky asset jumps, $k(t)$, instead, is a strict consequence of common shock dependence.
Moreover, we observe that
$$|u^{\b^*}(t, Y^\b_t)  -  u^{{\b,no}^*} (t)| \leq k(t)  | \bar w(t, Y^\b_t)|,$$
which provides an upper bound on the difference of reinsurance strategies with and without common shock, in terms of the investment strategy and the amplitude of jumps.

We finally see that if
$$(t, Y^\b_t) \in \left\{(t, y^\b) \in  [0,T]\times \R:\   \mu(t) - r(t)   > \lambda^\b(t, y^\b) k(t) (1+\theta_R^\b) \mathbb{E} [Z^\b] \right\},$$
then by \eqref{inv},
$| \bar w(t, Y^\b_t)| = \bar w(t, Y^\b_t)>0$, and hence $|u^{\b^*}(t, y^\b)  -  u^{{\b,no}^*} (t)| = u^{{\b,no}^*} (t) - u^{\b^*}(t, Y^\b_t)$ decreases when $ \lambda^\b(t, Y^\b_t)$ 
 increases.

If, instead,
$$(t, Y^\b_t) \in \left\{(t, y^\b) \in  [0,T]\times \R:\   \mu(t) - r(t)   <  \lambda^\b(t, y^\b) k(t) (1+\theta_R^\b) \mathbb{E} [Z^\b] \right\},$$
then $| \bar w(t, Y^\b_t)| = - \bar w(t, Y^\b_t)>0$, and hence $|u^{\b^*}(t, Y^\b_t)  -  u^{{\b,no}^*} (t)| = u^{\b^*}(t, Y^\b_t)  -  u^{{\b,no}^*} (t)$
increases when  $ \lambda^\b(t, Y^\b_t)$ and the reinsurance safety loading $\theta_R^\b$  decrease (because $u^{{\b,no}^*} (t)$ also increases  with  $\theta_R^\b$).

\section{Existence and uniqueness of classical solution to the HJB-equation}\label{sec:existence}

In this section we provide sufficient conditions  for existence and uniqueness of a classical solution $\psi(t,y^\a,y^\b)$ to the reduced HJB-equation  \eqref{eq:bellman1}.  We have seen that
this is implied by existence of classical solutions to PDEs \eqref{eq:bellman2} and \eqref{eq:bellman3}.

\begin{proposition}\label{P1}
Assume that the functions
\begin{itemize}
\item[(i)] $c^\a$, and  $\lambda^\a$  are bounded and Lipschitz-continuous in $(t,y^\a) \in [0,T] \times \R$;
\item[(ii)] $q^\a$ is  bounded and Lipschitz-continuous in $(t,y^\a) \in [0,T]\times \R$ uniformly with respect to $u^\a \in [0,1]$;
\item[(iii)]$a^\a$ is continuous with $a^\a(t)>\kappa>0$.
\end{itemize}
Then, there exists a unique solution $\psi_1(t,y^\a) \in C^{1,2}([0,T] \times \R)$ of the PDE  \eqref{eq:bellman2}.
\end{proposition}

\begin{proof}
This result has been proved in \citet[Corollary 8.3]{brachetta2019proportional}. It can also be derived from \citet[Theorem 5.3]{pham1998optimal}, under similar hypotheses.
\end{proof}

\begin{proposition}\label{P2}
Assume that the functions
\begin{itemize}
\item[(i)] $a^\b,b^\b,\mu, r,\sigma$ are continuous, with $a^\b(t), \sigma(t)>\kappa>0$;
\item[(ii)] the functions $c^\b, q^\b$ are continuous, Lipschitz in $(t,y^\b) \in [0,T]\times \R$, uniformly with respect to $u^\b \in [0,1]$;
\item[(iii)] $\lambda^\b$ is bounded.
\end{itemize}
Then, there exists a unique solution $\psi_2(t,y^\b) \in C^{1,2}([0,T] \times \R)$ of the PDE \eqref{eq:bellman3}.
\end{proposition}

\begin{proof}
The proof of this result follows, for instance from \citet[Theorem 5.3]{pham1998optimal} or from \citet[Theorem 1]{colaneri2019classical}.
\end{proof}

\section{Conclusions}
In this paper we have investigated the optimal investment and proportional reinsurance problem of an insurance company with exponential utility preferences in a stochastic-factor model where the stock market and the insurance market are correlated by means of a common shock.
The insurance company experiences two types of claims, namely ordinary claims and catastrophic claims, which correspond to two different lines of business, and their arrival intensities are affected by environmental (stochastic) factors, as well as the insurance and reinsurance premia; moreover, the arrival of catastrophic claims affects risky asset prices by inducing downward jumps.
The proposed model specification 
includes then two main features: firstly, a possible dependence between the financial and the pure insurance frameworks due to the common shock effect; secondly, a dependence of claim arrival intensities on some exogenous stochastic factors, and that allows to consider more general valuation premia.\\
Under suitable conditions on model coefficients, by applying the Hamilton-Jacobi-Bellman approach we have provided a classical solution of the resulting optimization problem  and characterized the optimal investment and reinsurance strategy. We have performed a comparison analysis on the optimal strategies with and without a common shock effect in the underlying financial-insurance setting. Our findings show that
in the common shock case, 
the insurance company 
invests a smaller quote in the financial market, compared to the optimal investment strategy without common shock, and the reinsurance strategy corresponding to the second line of business depends on the investment strategy.
In particular, when the impact of the common shock on the risky asset price is small, the insurance company takes a long position in the risky asset and the optimal reinsurance is smaller than in case of no-common shock. Conversely, when the impact of the common shock on the risky asset price is large the insurance company takes a short position in the risky assets and the optimal reinsurance is larger than in case of no-common shock. The optimal investment strategies with and without common shocks deviate more, as the intensity $\lambda^\b(t,Y^\b_t)$ of claim arrival increases and the size of jumps $K(t, Z^\b)$ of the risky asset due to the shock gets larger. This behavior is confirmed in case of classical valuation principles, e.g. the expected value principle, where in addition it is possible to explicitly quantify an upper bound for difference of reinsurance strategies with and without common shock, in terms of the investment strategy and the amplitude of asset price jumps.

\section*{Acknowledgements}
The authors  are members of INdAM-GNAMPA and their work has been partially supported through the Project U-UFMBAZ-2020-000791.

\appendix
\section{Proofs}\label{app:proofs}

This section collects a few technical proofs of the results stated in the paper.

\begin{proof}[Proof of Proposition \ref{ADM}]
We start with the condition \eqref{eq:integrab_strategy}. We have
\begin{align}
&\mathbb{E} \left[\int_0^T \left(|w_t| \left(|\mu(t) - r(t)|+ \lambda^\b(t, Y^\b_t) \right) + w^2_t \sigma^2(t) \right)\ud t \right] \\
&\leq \mathbb{E}\left[\int_0^T \eta(t)\left(|\mu(t) - r(t)|+ \delta(t) \right) + \eta^2(t) \sigma^2(t)\ud t\right],
\end{align}
which is finite because of conditions \eqref{w_ass}, \eqref{eq:delta1}  and $(i)$ in Assumption \ref{ASS}.
Condition \eqref{eq:integrab_strategy2} is implied by \eqref{eq:delta2}.
Next, we will prove that $\esp{e^{-\gamma X^\alpha_T}}<\infty$.
First, we observe that
$c^\a(t, Y^\a_t)+c^\b(t, Y^\b_t)>0$ and that $q^\a(t, Y^\a_t, u^\a_t)+q^\b(t, Y^\b_t, u^\b_t)\leq q^\a(t, Y^\a_t,0)+q^\b(t, Y^\b_t,0))<2\delta(t)$. Then, it holds that
\begin{align}
e^{-\gamma\int_0^TB(t,T) (c^\a(t, Y^\a_t)+c^\b(t, Y^\b_t)-q^\a(t, Y^\a_t, u^\a_t)-q^\b(t, Y^\b_t, u^\b_t))\ud t}\leq e^{2\gamma \overline{B}\int_0^T \delta(t)  \ud t}= k_0 <\infty.
\end{align}
Denote $k_1=k_0 e^{R_0\overline{B}}$ and $\varrho= 2\gamma \overline{B}$. From equation \eqref{eqn:wealth_sol} we have that
\begin{align}
&\mathbb{E}  \left[e^{-\gamma X^\alpha_T}\right]=e^{R_0\overline{B}} \mathbb{E}\big[ e^{-\gamma\int_0^TB(t,T) (c^\a(t, Y^\a_t)+c^\b(t, Y^\b_t)-q^\a(t, Y^\a_t, u^\a_t)-q^\b(t, Y^\b_t, u^\b_t))\ud t}\\
& \quad \times e^{-\gamma\int_0^TB(t,T)w_t\sigma(t)\left(\frac{\mu(t)-r(t)}{\sigma(t)}\ud t+\ud W_t \right)}e^{\gamma\int_0^TB(t,T)\int_0^{+\infty} z u^\a_t m^\a(\ud t, \ud z)}e^{\gamma\int_0^TB(t,T)\int_0^{+\infty} (z u^\b_t +w_t K(t,z))m^\b(\ud t, \ud z)}\big]\\
&\leq k_1 \mathbb{E} \big[e^{-\gamma\int_0^TB(t,T)w_t\sigma(t)\left(\frac{\mu(t)-r(t)}{\sigma(t)}\ud t+\ud W_t \right)}e^{\gamma\int_0^TB(t,T)\int_0^{+\infty} z m^\a(\ud t, \ud z)} e^{\gamma\int_0^TB(t,T)\int_0^{+\infty} (z + \eta(t))m^\b(\ud t, \ud z)}\big]\\
& \leq \frac{k_1}{2}\mathbb{E}\big[e^{-2\gamma\int_0^TB(t,T)w_t\sigma(t)\left(\frac{\mu(t)-r(t)}{\sigma(t)}\ud t+\ud W_t \right)}\big]\\
& \quad +\frac{k_1}{2} \mathbb{E}\big[e^{2\gamma\int_0^TB(t,T)\int_0^{+\infty} z  m^\a(\ud t, \ud z)}e^{2\gamma\int_0^TB(t,T)\int_0^{+\infty} (z + \eta(t))m^\b(\ud t, \ud z)}\big]\\
&= \frac{k_1}{2}\mathbb{E}\big[e^{-2\gamma\int_0^TB(t,T)w_t\sigma(t)\left(\frac{\mu(t)-r(t)}{\sigma(t)}\ud t+\ud W_t \right)}\big]\\
&\quad +\frac{k_1}{2} \mathbb{E}\big[e^{\varrho\int_0^T\int_0^{+\infty} z  m^\a(\ud t, \ud z)}\big]\mathbb{E}\big[e^{\varrho\int_0^T\int_0^{+\infty} (z + \eta(t))m^\b(\ud t, \ud z)}\big],\label{eq:inequalities1}
\end{align}
where in the first inequality we have used that $u^\a_t, u^\b_t$ and $K(t,z)$ are bounded from above by $1$ and that $w_t<\eta(t)$ for all $t \in [0,T]$, and the last equality follows from the independence of the random measures  $m^\a$ and $m^\b$ and the fact that $B(t,T)<\overline{B}$ for all $t \in [0,T]$.
We show now that  each of the three expected values in \eqref{eq:inequalities1} is finite.

\textbf{(i).} We begin with the first term
 \begin{align}
\mathbb{E}\big[e^{-2\gamma \int_0^TB(t,T)w_t\sigma(t)\left(\frac{\mu(t)-r(t)}{\sigma(t)}\ud t+\ud W_t \right)}\big]. \end{align}
Consider the probability measure $\Q$ with density process $\ds {\ud  \Q \over \ud \P} \Big{|}_{\F_t} = L_t$, for all $t \in [0,T]$, where
 $$L_t = e^{\ds- \int_0^t {\mu(s) - r(s) \over \sigma(s)} \ud   W_s-\frac{1}{2}\int_0^t\frac{(\mu(s) - r(s))^2}{\sigma^2(s)}\ud s}.$$
 Clearly, $L$ is a square integrable $\P$-martingale and the measure $\Q$ is equivalent to $\P$ (hence also $L^{-1}$ is a square integrable $\Q$-martingale). Let $W^\Q_t := \int_0^t {\mu(s) - r(s) \over \sigma(s)} \ud s + W_t$, for each $t \in [0,T]$; by Girsanov's theorem the process $W^\Q=\{W_t^\Q,\ t \in [0,T]\}$ is a $\Q$-Brownian motion. 
 Then, we have
 \begin{align}
 &\mathbb{E}\big[e^{-2\gamma \int_0^TB(t,T)w_t\sigma(t)\left(\frac{\mu(t)-r(t)}{\sigma(t)}\ud t+\ud W_t \right)}\big]
\leq  \mathbb{E}^\Q \left[  L^{-1}_T e^{ -  \int_0^T \varrho w_s \sigma(s) \ud W^\Q_s} \right]
  \\
 & \leq {1\over 2} \mathbb{E}^\Q [  L^{-2}_T]  + \frac{1}{2}\mathbb{E}^\Q \left[e^{ -  \int_0^T 2 \varrho w_s \sigma(s) \ud W^\Q_s}\right] \leq {1\over 2}  \mathbb{E}^\Q [  L^{-2}_T]  + \frac{1}{2}e^{  2 \varrho^2 \int_0^T  \eta^2(s) \sigma^2(s) \ud s }   <  \infty,
\end{align}
where the last inequality follows from condition $(iii)$ in Assumption \ref{ASS} and  \eqref{eq:delta1}.

\textbf{(ii).} Now, we consider the second expectation. Using the fact that $\{Z^\a_n\}_{n \in \mathbb{N}}$ are independent and identically distributed, independent of $N^\a$, and that $N^\a$ is a Poisson process with intensity  $\{\lambda^\a(t,Y_t^\a),\ t\in [0,T]\}$, we get
\begin{align}
&\mathbb{E}\big[e^{\varrho\int_0^T\int_0^{+\infty} z  m^\a(\ud t, \ud z)}\big]=\mathbb{E}\big[e^{\varrho\sum_{n=1}^{N^\a_T}Z^\a_n}\big]=\sum_{n\ge 0}\mathbb{E}\big[e^{\varrho\sum_{j=1}^{N^\a_T}Z^\a_j}|N^\a_T=n\big] \P(N^\a_T=n) \\&=
\sum_{n\ge 0}\left(\mathbb{E}\big[e^{\varrho Z^\a}\big]\right)^n \P(N^\a_T=n)=e^{(\Delta(T)-1)\int_0^T\lambda^\a(t, Y^\a_t) \ud t}\le e^{(\Delta(T)-1)\int_0^T \delta(t)\ud t},
\end{align}
which is finite since $\Delta(T):=\mathbb{E}\big[e^{\varrho Z^\a}\big]<\infty$, because of Assumption \ref{ASS}-$(ii)$, and $\int_0^T\lambda^\a(t, Y^\a_t)\ud t\leq \int_0^T\delta(t)\ud t< \infty$, because of Assumption \ref{ASS}-$(i)$.

\textbf{(iii).} Finally, we consider the last expectation.
%
We define the process $\Gamma=\{\Gamma_t,\ t \in [0,T]\}$ where $\Gamma_t := \int_0^t \int_0^{+\infty}  ( z+  \eta(s)) \,m^\b(\ud s, \ud z)$; then
\begin{align}
e^{\varrho  \Gamma_T} &= 1 + \sum_{s \leq T} \left( e^{\varrho \Gamma_{s}} -  e^{\varrho \Gamma_{s^-}}\right)= 1 + \int_0^T \int_0^{+\infty}  e^{\varrho \Gamma_{s^-}} \big(e^{\varrho (z + \eta(s))} -1\big)  \,m^\b(\ud s, \ud z).
\end{align}
Taking the expectation of both sides of equality yields
\begin{align}
\mathbb{E}[e^{\varrho \Gamma_T} ] & = 1 + \mathbb{E} \left[ \int_0^T \int_0^{+\infty}  e^{\varrho \Gamma_{s^-}}\big( e^{\varrho (z + \eta(s))} -1\big)  \lambda^\b(s, Y^\b_s) F^\b(\ud z) \ud s \right ]\\
&\leq  1 + \int_0^T  \mathbb{E} [e^{\varrho Z^\b}]  \mathbb{E} [ e^{\varrho \Gamma_{s}}] \delta(s) e^{\varrho  \eta(s)} \ud s.
\end{align}
By Gronwall's Lemma and condition \eqref{eq:delta2} we get
$$
\mathbb{E} [e^{\varrho \Gamma_T} ] \leq e^{  \mathbb{E} \big[ e^{\varrho Z^\b}\big]
\int_0^T \delta(s) e^{ \varrho \eta(s)} \ud s }< \infty,
$$
which concludes the proof.
\end{proof}

\begin{proof}[Proof of Theorem \ref{Verification}]
Let $\psi \in \mathcal{C}^{1,2,2}((0,T)\times \R^2)\cap \mathcal{C}([0,T]\times \R^2)$\footnote{$\mathcal{C}^{1,2,2}((0,T)\times \R^2)\cap \mathcal{C}([0,T]\times \R^2)$ indicates the set of functions $f(t, y^\a, y^\b)$ which is $\mathcal{C}^1$ in $t$ on $(0,T)$, and $\mathcal{C}^2$ in $(y^\a, y^\b)$ on $\R^2$ and jointly continuous in $(t, y^\a, y^\b)$ on $[0,T]\times \R^2$.} be a classical solution of the HJB equation \eqref{eq:bellman1} with the final condition \eqref{eq:final_condition1}. Then, the function $V$ defined in \eqref{eq:ansatz} solves the HJB problem in \eqref{eq:bellman}. Hence, for any
$(t,y^\a, y^\b, x) \in [0,T] \times \R \times \R \times \R$, we get that for all $s \in [t,T]$ and for all $\alpha \in \mathcal A$,
\begin{equation}
\frac{\partial V}{\partial t}\left(s,Y_{t,y^\a}^\a(s),Y_{t,y^\b}^\b(s),X_{t,x}^\alpha(s)\right) + \mathcal{L}^\alpha_t V\left(s,Y_{t,y^\a}^\a(s),Y_{t,y^\b}^\b(s),X_{t,x}^\alpha(s)\right) \ge 0,
\end{equation}
where $\{Y_{t,y^\a}^\a(s),\ s \in [t,T]\}$, $\{Y_{t,y^\b}^\b(s),\ s \in [t,T]\}$ and $\{X_{t,x}^\alpha(s),\ s \in [t,T]\}$ denote the solutions to \eqref{def:y1}, \eqref{def:y2} and \eqref{eqn:wealth_proc}, starting from $(t,y^\a) \in [0,T] \times \R$, $(t,y^\b) \in [0,T] \times \R$ and $(t,x) \in [0,T] \times \R$, respectively.
By applying It\^o's formula, we get
\begin{align}
& V\left(T,Y_{t,y^\a}^\a(T),Y_{t,y^\b}^\b(T),X_{t,x}^\alpha(T)\right) = V(t,y^\a, y^\b, x)\\
&   + \!\int_t^T\!\!\left[\frac{\partial V}{\partial t}\left(s,Y_{t,y^\a}^\a(s),Y_{t,y^\b}^\b(s),X_{t,x}^\alpha(s)\right)\! +\! \mathcal{L}^\alpha_t V\left(s,Y_{t,y^\a}^\a(s),Y_{t,y^\b}^\b(s),X_{t,x}^\alpha(s)\right)\right]\! \ud s \!+\! M_T, \quad {} \label{eq:v-ito}
\end{align}
with $\{M_r,\ r \in [t,T]\}$ defined by
\begin{align}
& M_r =\int_t^r a^\a(s)\frac{\partial V}{\partial y^\a}\left(s,Y_s^\a,Y_s^\b,X_s^\alpha\right)\ud W_s^\a  + \int_t^r a^\b(s)\frac{\partial V}{\partial y^\b}\left(s,Y_s^\a,Y_s^\b,X_s^\alpha\right)\ud W_s^\b \\
& \ + \int_t^r w_s \sigma(s)\frac{\partial V}{\partial x}\left(s,Y_s^\a,Y_s^\b,X_s^\alpha\right)\ud W_s\\
& \ + \int_t^r \int_0^{+\infty} \left[V\left(s,Y_s^\a,Y_s^\b,X_s^\alpha-zu_s^\a\right) - V\left(s,Y_s^\a,Y_s^\b,X_s^\alpha\right)\right]\widetilde m^\a(\ud s, \ud z)\\
& \ + \int_t^r \int_0^{+\infty} \left[V\left(s,Y_s^\a,Y_s^\b,X_s^\alpha-zu_s^\b-w_sK(s,z)\right) - V\left(s,Y_s^\a,Y_s^\b,X_s^\alpha\right)\right]\widetilde m^\b(\ud s, \ud z).
\end{align}
where for $i=1,2$, $\widetilde m^\i(\ud s, \ud z):=m^\i(\ud s, \ud z) - \lambda^\i(s,Y_s^\i)F^\i(\ud z)\ud s$.
Next, we show  that $\{M_r,\ r \in [t,T]\}$ is an $\bF$-local martingale.
Let us define a sequence of random times $\{\tau_n\}_{n \in \mathbb N}$ as follows:
\begin{equation}
\tau_n:=\inf\{s \in [t,T]|\ X_s^\alpha < -n \vee |Y_s^\a| > n \vee |Y_s^\b| > n\}.
\end{equation}
In the sequel we denote by $C_n$ any nonnegative constant that depends on $n$. Then, we have that for $i=1,2,$
\begin{align}
& \esp{\int_t^{T \wedge \tau_n}\left(a^\i(s)\frac{\partial V}{\partial y^\a}\left(s,Y_s^\a,Y_s^\b,X_s^\alpha\right)\right)^2 \ud s }\\
& = \esp{\int_t^{T \wedge \tau_n}\left(a^\i(s)e^{ - \gamma X_s^\alpha B(s,T)}  \frac{\partial \psi}{\partial y^\i} (s,Y_s^\a,Y_s^\b)\right)^2 \ud s}
\leq C_n \esp{\int_t^{T \wedge \tau_n} (a^\i(s))^2\ud s} < \infty,
\end{align}
for all $n \in \mathbb{N}$ where the last inequality follows from condition \eqref{eq:a_i}.
Further, we have that
\begin{align}
& \esp{\int_t^{T \wedge \tau_n}\!\!\left(w_s \sigma(s)\frac{\partial V}{\partial x}\left(s,Y_s^\a,Y_s^\b,X_s^\alpha\right)\right)^2 \ud s }\\
& = -\esp{\int_t^{T \wedge \tau_n}\!\!\!\left(w_s \sigma(s)\gamma B(s,T)e^{- \gamma X_s^\alpha B(s,T) } \psi(s,Y_s^\a,Y_s^\b)\right)^2 \ud s }\leq C_n \esp{\int_t^{T \wedge \tau_n}w^2_s \sigma^2(s)\ud s } < \infty,
\end{align}
for all $n \in \mathbb N$, where the last inequality holds because $\{w_t,\ t \in [0,T]\}$ is admissible. In addition, for the first jump term we get
\begin{align}
& \esp{\int_t^{T \wedge \tau_n}\int_0^{+\infty}\left|V\left(s,Y_s^\a,Y_s^\b,X_s^\alpha-zu_s^\a\right) - V\left(s,Y_s^\a,Y_s^\b,X_s^\alpha\right)\right|\lambda^\a(s,Y_s^\a)F^\a(\ud z)\ud s} \\
& = \mathbb E\biggl[\int_t^{T \wedge \tau_n}\int_0^{+\infty}e^{ - \gamma X_s^\alpha B(s,T)}  |\psi (s,Y_s^\a,Y_s^\b)|
\left|e^{\gamma zu^\a_s B(s,T)}- 1\right|\lambda^\a(s,Y_s^\a)F^\a(\ud z)\ud s\biggr]\\
& \leq C_n \esp{\int_t^{T \wedge \tau_n}\int_0^{+\infty}e^{\gamma \overline{B} z }\lambda^\a(s,Y_s^\a)F^\a(\ud z)\ud s}\\
& \leq C_n \esp{e^{\gamma \overline{B}Z^\a }}\esp{\int_t^{T \wedge \tau_n}\lambda^\a(s,Y_s^\a) \ud s}< \infty,\quad n \in \mathbb N,
\end{align}
thanks to Assumption \ref{ASS}. Finally, for the second jump term we obtain
\begin{align}
& \mathbb{E}\bigg[ \int_t^{T \wedge \tau_n}\!\!\!\int_0^{+\infty}\!\!\!\left|V\left(s,Y_s^\a,Y_s^\b,X_s^\alpha-zu_s^\b-w_sK(s,z)\right)\! -\! V\left(s,Y_s^\a,Y_s^\b,X_s^\alpha\right)\right|\\
&\quad\times  \lambda^\b(s,Y_s^\b)F^\b(\ud z)\ud s\bigg] \\
& = \mathbb E\biggl[\int_t^{T \wedge \tau_n}\!\!\!\int_0^{+\infty}\!\!\!e^{ - \gamma X_s^\alpha B(s,T)}  |\psi (s,Y_s^\a,Y_s^\b)| \left|e^{\gamma B(s,T) (zu_s^\b + w_sK(s,z)) }- 1\right|\lambda^\b(s,Y_s^\b)F^\b(\ud z)\ud s\biggr]\\
& \leq C_n\esp{\int_t^{T \wedge \tau_n} \int_0^{+\infty} e^{\gamma \overline{B} (z + |w_s|) }\lambda^\b(s,Y_s^\b)F^\b(\ud z)\ud s}< \infty,
\end{align}
for all $n \in \mathbb{N}$, where the last inequality follows from condition \eqref{eq:integrab_strategy2}. Thus, $\{M_r,\ r \in [t,T]\}$ is an $\bF$-local martingale with localizing sequence $\{\tau_n\}_{n \in \mathbb N}$.
We take the conditional expectation of \eqref{eq:v-ito} with $T \wedge \tau_n$ in place of $T$, and we obtain that
\begin{equation}
\begin{split}
& \esp{V\left(T\wedge \tau_n,Y_{t,y^\a}^\a(T\wedge \tau_n),Y_{t,y^\b}^\b(T\wedge \tau_n),X_{t,x}^\alpha(T\wedge \tau_n)\right)|\F_t} \ge V(t,y^\a, y^\b, x),
\end{split}
\end{equation}
for any $\alpha \in \mathcal A$, $(t,y^\a,y^\b,x) \in \llbracket0,T\wedge \tau_n\rrbracket \times \R^3$, $n \in \mathbb N$. Note that
\begin{equation}
\begin{split}
&\esp{V\left(T\wedge \tau_n,Y_{t,y^\a}^\a(T\wedge \tau_n),Y_{t,y^\b}^\b(T\wedge \tau_n),X_{t,x}^\alpha(T\wedge \tau_n)\right)^2} \\
& \quad= \esp{e^{- 2\gamma X_{t,x}^\alpha(T\wedge \tau_n) B(T \wedge \tau_n, T)}  \psi (T \wedge \tau_n,Y_{t,y^\a}^\a(T\wedge \tau_n),Y_{t,y^\b}^\b(T\wedge \tau_n))^2}<\infty.
\end{split}
\end{equation}
 Then,  $\{V\left(T\wedge \tau_n,Y_{t,y^\a}^\a(T\wedge \tau_n),Y_{t,y^\b}^\b(T\wedge \tau_n),X_{t,x}^\alpha(T\wedge \tau_n)\right)\}_{n \in \mathbb N}$ is a family of uniformly integrable random variables and converges almost
surely. Using the monotonicity and the boundedness of the sequence $\{\tau_n\}_{n \in \mathbb N}$ and the fact that the processes $\{Y_{t,y^\a}^\a(s),\ s \in [t,T]\}$, $\{Y_{t,y^\b}^\b(s),\ s \in [t,T]\}$ and $\{X_{t,x}^\alpha(s),\ s \in [t,T]\}$, see \eqref{def:y1}, \eqref{def:y2} and   \eqref{eqn:wealth_sol}, taking the limit for $n \to +\infty$, we get
\begin{equation}
\begin{split}
& \esp{V\left(T,Y_{t,y^\a}^\a(T),Y_{t,y^\b}^\b(T),X_{t,x}^\alpha(T)\right)\Big{|}\F_t}\\
& = \lim_{n \to +\infty}\esp{V\left(T\wedge \tau_n,Y_{t,y^\a}^\a(T\wedge \tau_n),Y_{t,y^\b}^\b(T\wedge \tau_n),X_{t,x}^\alpha(T\wedge \tau_n)\right)\Big{|}\F_t}\ge V(t,y^\a, y^\b, x),
\end{split}
\end{equation}
for every $\alpha \in \A$, $t \in [0,T]$. Recall that $V(T, y^\a, y^\b, x)=e^{-\gamma x}$ (see equation \eqref{eq:final_condition}), then from the previous inequality we get
\begin{equation}\label{eq:ver1}
\esp{e^{-\gamma X_{t,x}^\alpha(T)}} \ge V(t,y^\a, y^\b, x), \quad \alpha \in \A, \quad t \in [0,T].
\end{equation}
Since $u^{\a,*}(t,y^\a)$ and  $(u^{\b,*}(t,y^\b), w^*(t,y^\b))$ are minimizers of  $\inf_{u^\a \in [0,1]} \Psi_1(t,y^\a,u^\b)$, and
 $\inf_{u^\b \in [0,1], w \in \R} \Psi_2(t,y^\b,u^\b,w) $, respectively, then
\begin{equation}
\frac{\partial V}{\partial t}(t,y^\a,y^\b,x)+\mathcal L_t^{\alpha^*}V(t,y^\a,y^\b,x)=0,
\end{equation}
with $\alpha^*(t,y^\a,y^\b)=(u^{\a,*}(t,y^\a),u^{\b,*}(t,y^\b), w^*(t,y^\b))$. Replicating the computations above, replacing $\mathcal L^{\alpha}$ with $\mathcal L^{\alpha^*}$, we get the equality
\begin{equation}
{\mathbb{E}^{t,y^\a,y^\b,x}\bigl[ e^{-\gamma X_T^{\alpha^*}} \bigr]}=\inf_{\alpha \in\mathcal{A}}{\mathbb{E}^{t,y^\a,y^\b,x}\bigl[ e^{-\gamma X_T^\alpha} \bigr]}=V(t,y^\a,y^\b,x),
\end{equation}
and hence $\alpha^*_t = ({u^\a}^*(t, Y^\a_t), {u^\b}^*(t,Y^\b_t), w^*(t,Y^\b_t))$ is an optimal control.
\end{proof}

\begin{proof}[Proof of Lemma \ref{lemma:minimizer}]
We separate the proof of the lemma in four steps.

\textbf{Step 1 (Existence of $\widetilde u^\b$).} For any $(t, y^\b, w) \in [0,T] \times \R^2$, the function $H(t,w, y^\b, u^\b)$ defined in equation \eqref{H}, is continuous and increasing in $u^\b \in \R$. Then, by the implicit function theorem under condition \eqref{H1}
there is a unique measurable function $\tilde u^{\b}(t, y^\b, w)$ which solves equation \eqref{deriv}, i.e. $H(t,y^\b, \widetilde u^\b(t, y^\b, w), w) =0$.
Moreover, the function $H(t,y^\b, u^\b, w)$  is also continuous and in $w \in \R$. Hence $\tilde u^{\b}(t, y^\b, w
)$ is  also continuous  and it is decreasing in $w \in \R$.

\textbf{Step 2 (Existence of $\widetilde w$).}
Consider now the function $\widetilde H (t,y, u, w)$ defined in equation \eqref{eq:Htilde}.
For any  $(t,y^\b, u^\b) \in [0,T] \times \R^2$, $\widetilde H(t,y^\b, u^\b, w)$ is continuous in $w \in \R$ and satisfies
  \begin{equation}
\lim_{w \to  - \infty} \widetilde H (t,y^\b, u^\b, w)  = - \infty, \quad  \lim_{w \to  + \infty} \widetilde H(t,y^\b, u^\b, w) = + \infty,
\end{equation}
for all $(t, y^\b, u^\b) \in [0,T] \times \R^2$.
Hence, again by the implicit function theorem, there is a measurable function $\widetilde w(t, y^\b, u^\b)$ which satisfies
equation \eqref{derivgen}, that is, $ \widetilde H  (t,  y^\b, u^\b, \widetilde w(t, y^\b, u^\b)) = 0$ for all $(t, y^\b, u^2)\in [0,T]\times \R^2$.

Since $ \widetilde H(t,y^\b, u^\b, w) $ is increasing on $u^\b$, it holds that $\widetilde w(t,y^\b, u^\b)$ is monotonic decreasing in $u^\b$, for every  $(t, y^\b) \in  [0,T] \times \R$.

\textbf{Step 3 (Existence of $\bar w$).} Next, we show that the following limits hold:
\begin{align}
\lim_{w \to  - \infty} \widetilde H (t,y^\b, \tilde u^{\b}(t, y^\b, w
), w)  = - \infty, \label{eq:lim1}\\
 \lim_{w \to  + \infty} \widetilde H(t,y^\b, \tilde u^{\b}(t, y^\b, w), w) = + \infty, \label{eq:lim2}
\end{align}
for all $(t, y^\b) \in [0,T] \times \R$. Since the second limit is easily verified, we only show \eqref{eq:lim2}.
Considering the expression in equation \eqref{eq:Htilde} we see that
\[
\lim_{w \to  - \infty} \gamma B(t,T) \sigma^2(t) w -\! (\mu(t) - r(t) ) =-\infty.
\]
Now, we prove that
\begin{equation}\label{eq:lim3}
\lim_{w\to -\infty}  \lambda^\b(t, y^\b) \int_0^{+\infty} K(t,z) e^{\gamma B(t,T)( zu^\b(t,y^\b,w) + w K(t,z)) } F^\b(\ud z)=c_1<\infty.
\end{equation}
This allows us to conclude that \eqref{eq:lim1} holds.
Equation \eqref{eq:lim3} is satisfied if
\[\lim_{w \to  - \infty} \tilde u^{\b}(t, y^\b, w)z + K(t,z) w < \infty.\]
We prove the latter limit by contradiction. Suppose that
\[\lim_{w \to  - \infty} \tilde u^{\b}(t, y^\b, w)z + K(t,z) w = + \infty.\]
Then, since $H(t, y^\b, \widetilde u^\b(t,y^\b,w),w)=0$ (see Step 1), we also have that
\[\lim_{w \to  - \infty}H(t, y^\b, \widetilde u^\b(t,y^\b,w
),w)=0,
\]
equivalently,
\[
\lim_{u^\b \to  + \infty} \frac{\partial  q^\b}{\partial {u^\b}}(t,y^\b, u^\b)=-\infty,
\]
where the limit on the right side is taken for $u^\b\to +\infty$ because $\widetilde u^\b$ is decreasing in $w$.
This is a contradiction because the function $q^\b$ is increasing in $u^\b$ and hence $\frac{\partial  q^\b}{\partial {u^\b}}(t,y^\b, u^\b)>0$.

The conditions given by the limits \eqref{eq:lim1} and \eqref{eq:lim2}, together with the fact $ \tilde u^\b(t, y^\b, w)$ is continuous in $w$, imply that (again be the implicit function theorem) there exists a unique  $\bar w(t, y^\b)$ such that
\[ \widetilde H (t,y^\b, \tilde u^{\b}(t, y^\b, w), w) = 0,\]
for all  $(t, y^\b) \in  [0,T] \times \R$.


\textbf{Step 4 (The solution $(\bar u^\b(t, y^\b), \bar w(t, y^\b))$ of the system \eqref{deriv}--\eqref{derivgen}).}
Set $\bar u^{\b}(t, y^\b) = \tilde u^{\b}(t, y^\b, \bar w(t, y^\b))$. By Step 1--Step 3 we get that the pair $(\bar u^{\b}(t, y^\b), \bar w(t, y^\b))$ satisfies \eqref{deriv} and \eqref{derivgen}.

\end{proof}

\begin{proof}[Proof of Proposition \ref{stime}]
We start by showing \eqref{dis_1}.
Since for every $(t,y^\b) \in [0,T] \times \R$, $w(t, y^\b, u^\b)$ is monotonic decreasing in $u^\b$ (see the formula \eqref{monotonia}), we only need to show that 
\begin{equation}\label{eq:inequality}\widetilde w(t, y^\b, 0) < \frac{\mu(t) - r(t)}{\gamma B(t,T) \sigma^2(t)}, \end{equation}
for all $(t,y^\b)  \in  [0,T] \times \R$. Using the form of the function $\widetilde H$ in \eqref{eq:Htilde} and the fact that for all $(t, y^\b)\in [0,T]\times \R$, $\widetilde w(t, y^\b, 0)$ is the unique solution of the equation \eqref{derivgen} with $u^\b=0$, we get that (we omit the dependence on $(t, y^\b, 0)$ in $\widetilde w$ in the formula below)
\begin{align}
0=& \widetilde H(t,y^\b, 0, \widetilde w) \\
= &\gamma B(t,T) \sigma^2(t) \widetilde w - (\mu(t) - r(t) )  +   \lambda^\b(t, y^\b)\!\! \int_0^{+\infty} \!\!\!\!K(t,z) e^{\gamma B(t,T)\widetilde w K(t,z)} F^\b(\ud z) \\
 > & \gamma B(t,T) \sigma^2(t) \widetilde w - (\mu(t) - r(t) )
 \end{align}
for all $(t,y^\b, w)  \in  [0,T] \times \R^2$, which implies \eqref{eq:inequality}.

Now, we prove \eqref{dis_2}. Since for all $(t,y^\b) \in [0,T] \times \R$ the function $w(t, y^\b, u^\b)$ is monotonic decreasing in $u^\b$ (see \eqref{monotonia}), we will show that
\begin{align}\widetilde w(t, y^\b, 1) \geq \min \left\{0, \frac{\mu(t) - r(t)}{\gamma B(t,T) \sigma^2(t)} - \frac{  \delta(t) \int_0^{+\infty} e^{\gamma B(t,T) z } F^\b(\ud z)}{\gamma B(t,T) \sigma^2(t)} \right\},
\end{align}
 for all $(t,y^\b)   \in  [0,T] \times \R$. Since $(t, y^\b)\in [0,T]\times \R$,  the function $\widetilde w(t, y^\b, 1) $ is the unique solution of equation \eqref{derivgen} with $u^\b=1$ and we get the following two cases.
 \begin{itemize}
\item[(i)] If $(t, y^\b) \in C_2$, it holds that $\widetilde H(t,y^\b, 0,1) < 0$, and hence $\widetilde w(t, y^\b, 1) >0$.
\item[(ii)] Let $C^c_2$ be the complementary set of $C_2$.  If $(t, y^\b) \in C^c_2$, we have that $\widetilde H(t,y^\b, 0,1) > 0$, and hence $\widetilde w(t, y^\b, 1) < 0$, and for all $ w <0$ it holds that
\begin{align}
0&= \widetilde H(t,y^\b, 1,w ) \\
&\leq  \gamma B(t,T) \sigma^2(t) w - (\mu(t) - r(t) )  +   \lambda^\b(t, y^\b) \int_0^{+\infty} e^{\gamma B(t,T) z } F^\b(\ud z)\\
&\leq \gamma B(t,T) \sigma^2(t) w - (\mu(t) - r(t) )  + \delta(t)\int_0^{+\infty} e^{\gamma B(t,T) z } F^\b(\ud z).
\end{align}
This implies that  $\widetilde w(t, y^\b, 1) \geq  \frac{\mu(t) - r(t)}{\gamma B(t,T) \sigma^2(t)} - \frac{  \delta(t) \int_0^{+\infty} e^{\gamma B(t,T) z } F^\b(\ud z)}{\gamma B(t,T) \sigma^2(t)}$, for all $(t, y^\b) \in C^c_2$.
\end{itemize}
To prove the final assertion we first observe that if $(t,y^\b)   \in C_1$, then $\widetilde H(t,y^\b, 0, 0)>0$ and hence $\widetilde w(t, y^\b, 0) <0$. Since $\widetilde w$ is monotonic in $u^\b\in [0,1]$, from the form of the optimal reinsurance strategy \eqref{reins2}, we get that
$w^*(t,y^\b) < 0$, for all $(t,y^\b)   \in C_1$,
and using the fact that $\widetilde w(t, y^\b, 1) >0$ for all  $(t,y^\b) \in C_2$ we have that
$w^*(t,y^\b) > 0$, for all $(t,y^\b)   \in C_2,$
and this concludes the proof.
 \end{proof}

\bibliographystyle{plainnat}
\bibliography{bibliography}

\end{document}